%% file: main.tex
\pdfoutput=1
\documentclass{article}

\usepackage{arxiv}

\usepackage[T1]{fontenc}
\usepackage[utf8]{inputenc}

\usepackage{amsmath,amssymb,amsthm,bm,mathtools}
\usepackage{booktabs}
\usepackage{graphicx}
\usepackage{subcaption}
\usepackage{algorithm}
\usepackage{algpseudocode}
\usepackage{enumitem}

\usepackage{microtype}
\usepackage{url}
\usepackage{doi}
\usepackage{natbib}

\providecommand{\texorpdfstring}[2]{#1}

\theoremstyle{plain}
\newtheorem{theorem}{Theorem}[section]
\newtheorem{lemma}[theorem]{Lemma}
\newtheorem{proposition}[theorem]{Proposition}

\newtheorem{assumption}[theorem]{Assumption}

\theoremstyle{definition}
\newtheorem{definition}[theorem]{Definition}

\theoremstyle{remark}
\newtheorem{remark}[theorem]{Remark}

\title{Emissions-Robust Portfolios}

\author{Khizar Qureshi\\
Department of Mathematics\\
  Massachusetts Institute of Technology\\
  Cambridge, MA 02139 \\
  \texttt{kqureshi@mit.edu} \\
	\And
  H. Oliver Gao \\
  School of Civil and Environmental Engineering \\
  Cornell University \\
  Ithaca, NY 14850 \\
  \texttt{hg55@cornell.edu}}

\renewcommand{\headeright}{Emissions-Robust Portfolios}
\renewcommand{\undertitle}{}

\begin{document}
\maketitle

\begin{abstract}
We study portfolio choice when firm-level emissions intensities are measured with error. We introduce a scope-specific penalty operator that rescales asset payoffs as a smooth function of revenue-normalized emissions intensity. Under payoff homogeneity, unit-scale invariance, mixture linearity, and a curvature semigroup axiom, the operator is unique and has the closed form $P^{(m)}_j(r,\lambda)=\bigl(1-\lambda/\lambda_{\max,j}\bigr)^m r$. Combining this operator with norm- and moment-constrained ambiguity sets yields robust mean--variance and CVaR programs with exact linear and second-order cone reformulations and economically interpretable dual variables. In a U.S. large-cap equity universe with monthly rebalancing and uniform transaction costs, the resulting strategy reduces average Scope~1 emissions intensity by roughly 92\% relative to equal weight while exhibiting no statistically detectable reduction in the Sharpe ratio under block-bootstrap inference and no statistically detectable change in average returns under HAC inference. We report the return--emissions Pareto frontier, sensitivity to robustness and turnover constraints, and uncertainty propagation from multiple imputation of emissions disclosures.
\end{abstract}


\section{Introduction}\label{sec:intro}

Climate change is a canonical global stock externality: the damage-relevant state variable is the atmospheric concentration of greenhouse gases, not the contemporaneous flow of emissions. Anthropogenic CO$_2$ concentrations exceeded 420 ppm in 2023, relative to a pre-industrial baseline of roughly 280 ppm \citep{IPCC_AR6_SYR_2023}. A back-of-the-envelope calculation links temperature targets to the residual global carbon budget. Let $B_{\max}$ denote the cumulative CO$_2$ emissions compatible with a maximum admissible warming $\Delta T_{\max}$, and let $\kappa$ be the temperature response per gigaton of CO$_2$. Then
\begin{equation}
\begin{aligned}
B_{\max}
&= \frac{\Delta T_{\max}}{\kappa} \\
&\approx \frac{1.5^{\circ}\mathrm{C}}{4.8\times 10^{-4}\, ^{\circ}\mathrm{C}/\mathrm{GtCO}_2} \\
&\approx 310~\mathrm{GtCO}_2.
\end{aligned}
\label{eq:carbon-budget}
\end{equation}
At current emissions rates of roughly $40\,\mathrm{GtCO}_2/\text{year}$, this residual budget would be exhausted within a decade.

\begin{figure}[t]
  \centering
  \begin{subfigure}[t]{0.48\textwidth}
    \centering
    \includegraphics[width=\linewidth]{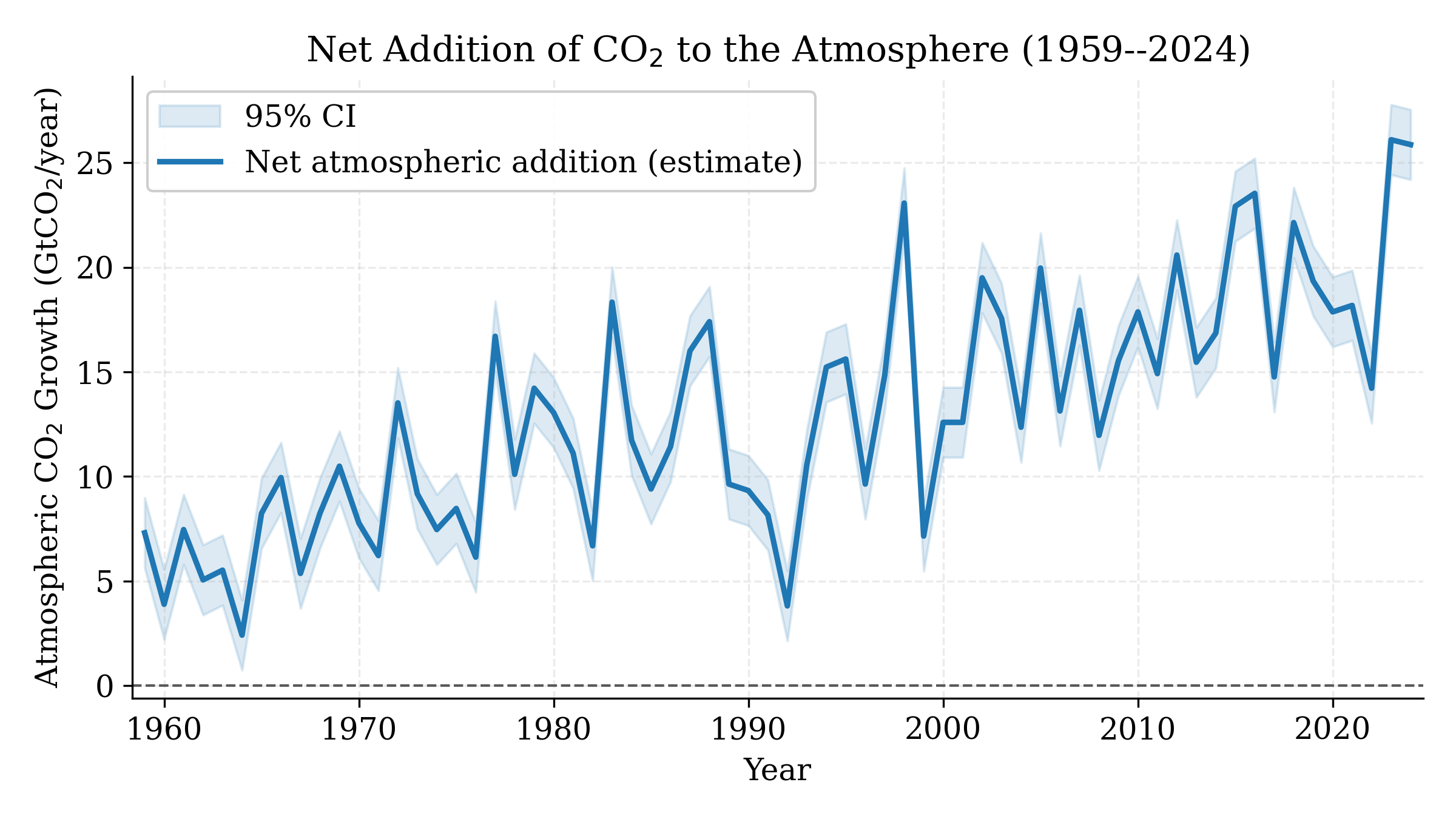}
    \caption{Net atmospheric addition of CO$_2$ (GtCO$_2$/yr), 1959--2024, with a 95\% confidence band. The band propagates NOAA Mauna Loa growth-rate uncertainty ($\pm 1.96\times 0.11$\,ppm) into mass units (1\,ppm $\approx$ 7.77\,GtCO$_2$).}
    \label{fig:atm_addition_ci}
  \end{subfigure}
  \hfill
  \begin{subfigure}[t]{0.48\textwidth}
    \centering
    \includegraphics[width=\linewidth]{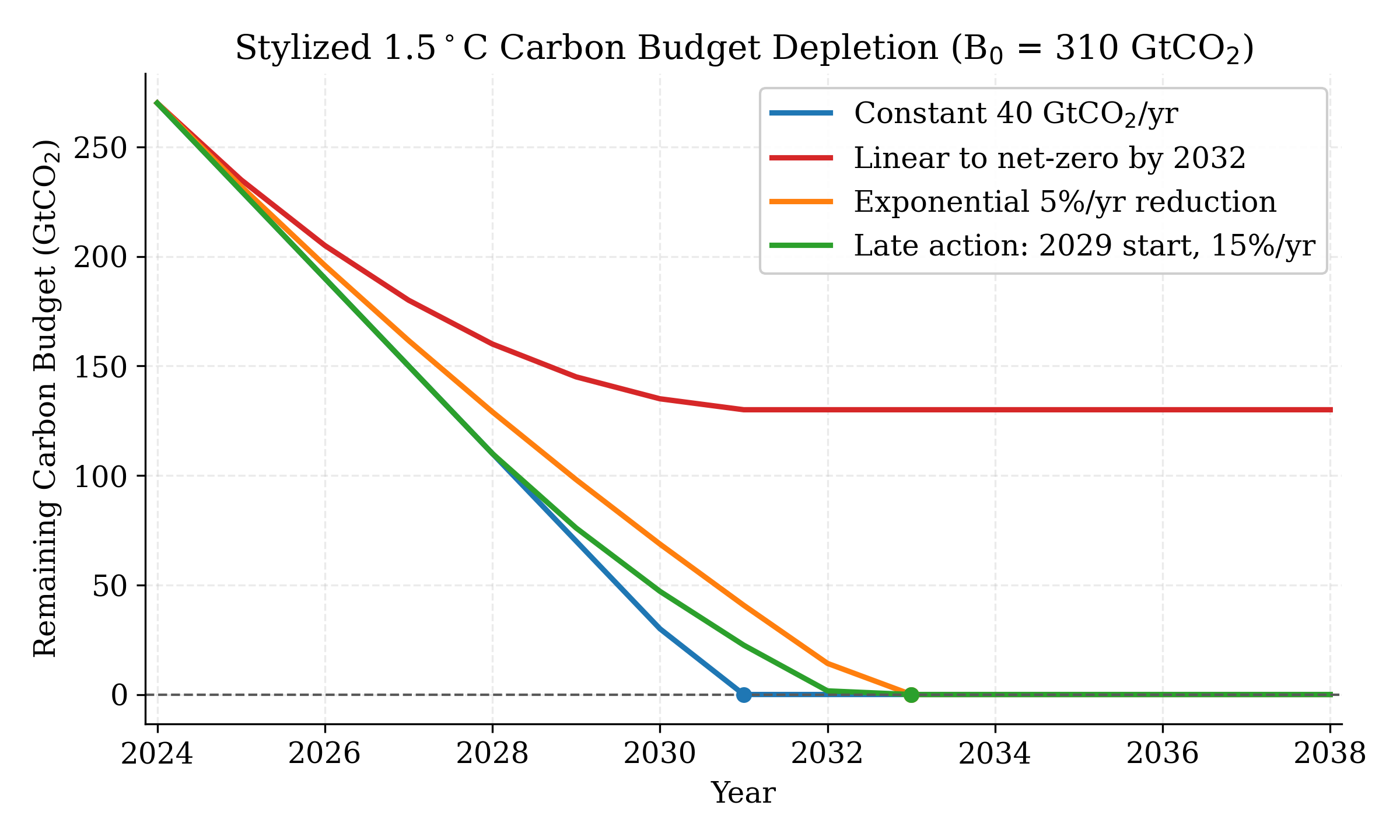}
    \caption{Stylized 1.5$^\circ$C remaining budget, $B_0=310$\,GtCO$_2$, under alternative usage paths: constant 40\,GtCO$_2$/yr, linear decline to net-zero by 2032, exponential 5\%/yr reduction. And late action (flat to 2028 then 15\%/yr decline). Exhaustion years are annotated where applicable.}
    \label{fig:budget_curves}
  \end{subfigure}
  \caption{\textbf{From atmospheric accumulation to policy-relevant budgets.}
Panel~(a) translates ppm growth into GtCO$_2$/yr and shows measurement uncertainty, panel~(b) illustrates how non-linear emissions pathways induce different curvature in the remaining-budget trajectory. Values are illustrative and sensitive to methodological choices (e.g., non-CO$_2$ forcers and probability thresholds).}
  \label{fig:atm_and_budget_side_by_side}
\end{figure}

These numbers underscore a simple point: anthropogenic warming is driven by the cumulative stock of CO$_2$. The IPCC Sixth Assessment Report (AR6) formalizes the remaining carbon quotas consistent with $1.5^{\circ}\mathrm{C}$ and $2^{\circ}\mathrm{C}$ pathways and emphasizes the need for deep, near-term mitigation \citet{IPCC_AR6_SYR_2023}. Flow estimates from the Global Carbon Budget \citep{Friedlingstein_2025_GCB2024} document record or near-record fossil emissions, highlighting the gap between physical constraints and realized trajectories. The Paris Agreement's Article~2.1(c) explicitly commits signatories to ``making finance flows consistent with a pathway towards low greenhouse gas emissions and climate-resilient development'' \citep{UNFCCC_Paris_2015}. Complementary analyses by the UNEP Emissions Gap Report \citep{UNEP_EGR_2023} quantify the shortfall between pledged policies and emissions pathways compatible with $1.5^{\circ}\mathrm{C}$. Together, these developments motivate a quantitative finance response: design portfolio frameworks that deliver verifiable emissions reductions without sacrificing financial performance or risk discipline.

Institutional investors sit at the transmission channel between these climate constraints and the real economy. Global capital markets intermediate on the order of \$400 trillion in financial assets \citep{Bolton2020}. Redirecting even a modest fraction of this capital toward climate-compatible firms could generate first-order mitigation effects. Two frictions, however, impede such realignment. First, corporate greenhouse gas (GHG) disclosures are incomplete, noisy, and subject to strategic reporting and methodological disagreement across data vendors \citep{Berg2022Rating}. Second, asset managers are still evaluated through classical lenses---risk-adjusted return, drawdown control, and turnover. Credible portfolio design must therefore integrate two classes of uncertainty on equal footing: (i) estimation risk in corporate GHG disclosures and (ii) the familiar mean--variance trade-off governing financial optimization \citep{Markowitz1952}.
In \S\ref{sec:results} we make this comparison concrete by benchmarking EAPO against deterministic inverse-intensity tilts that treat estimated intensities as known.

\begin{figure}[t]
  \centering
  \begin{subfigure}[t]{0.48\textwidth}
    \centering
    \includegraphics[width=\linewidth]{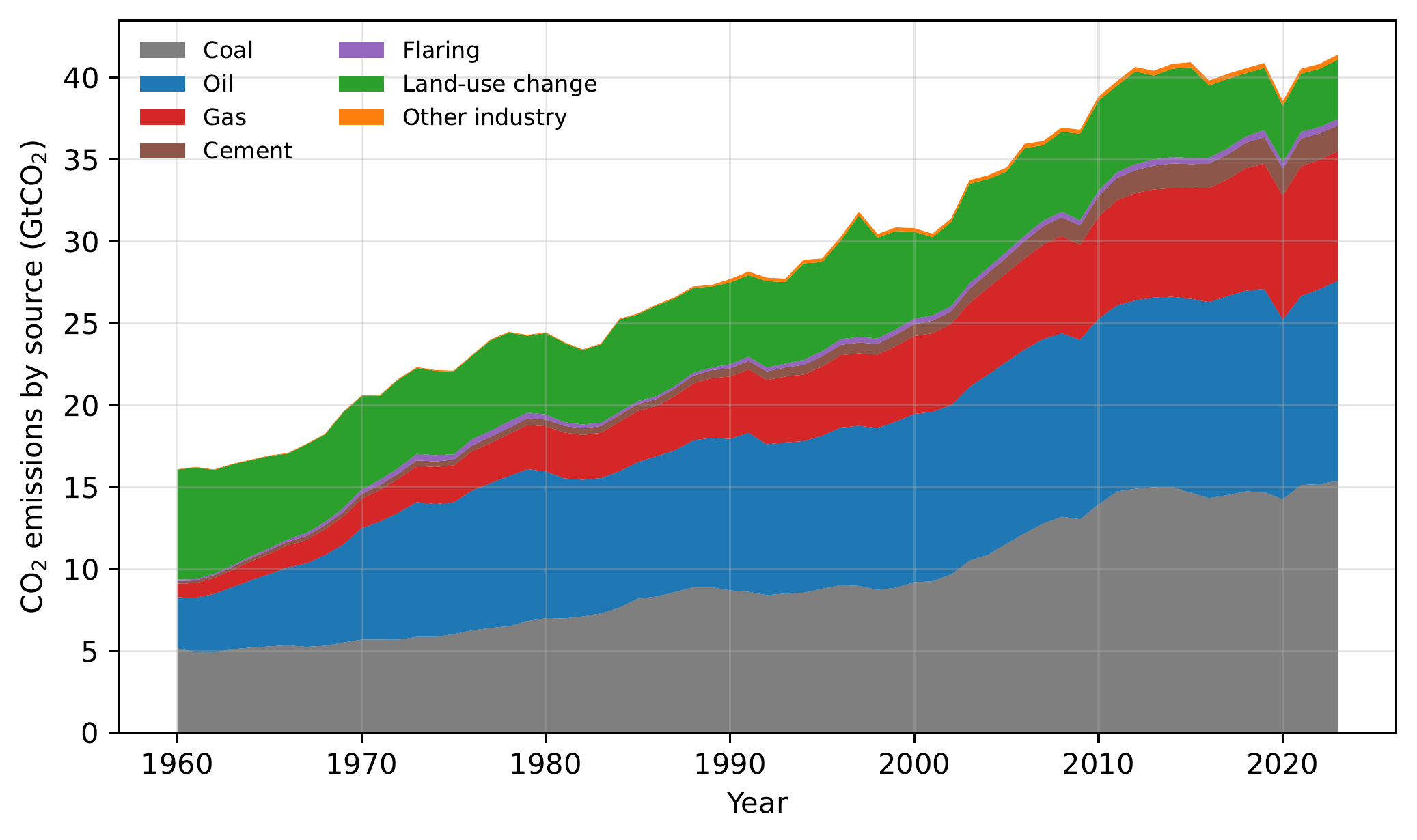}
    \caption{Global fossil and land-use CO$_2$ emissions by source, 1960--2023.}
    \label{fig:global_co2_by_source}
  \end{subfigure}
  \hfill
  \begin{subfigure}[t]{0.48\textwidth}
    \centering
    \includegraphics[width=\linewidth]{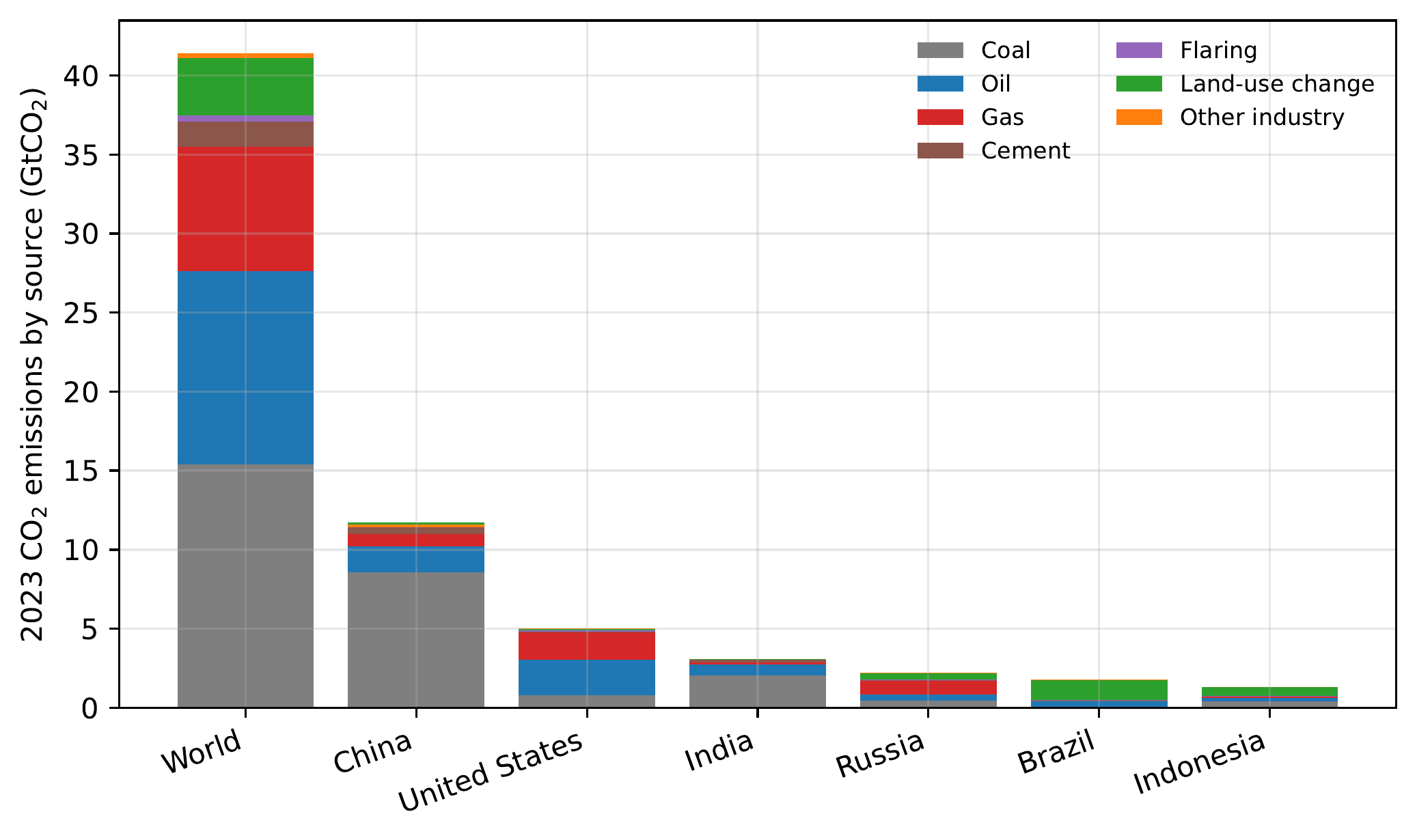}
    \caption{Decomposition of 2023 CO$_2$ emissions by source for major emitters.}
    \label{fig:co2_by_source_top_emitters_2023}
  \end{subfigure}
  \caption{Global and 2023 CO$_2$ emissions by source. Source: Our World in Data, ``CO$_2$ and Greenhouse Gas Emissions'' dataset.}
  \label{fig:co2_sources_combined}
\end{figure}

\paragraph{Literature and conceptual foundations.}

Debate on the financial performance of ESG-tilted portfolios dates back to early comparisons of socially responsible and conventional mutual funds \citep{Renneboog2008,Fabozzi2008}. Meta-analyses such as \citet{Friede2015} find on average a weakly positive relation between ESG characteristics and financial performance. Production-based asset pricing models attribute ESG premia to profitability and investment channels \citep{Balvers2017}. Building on these findings, a growing literature studies the optimal construction of ESG-integrated and low-carbon portfolios.

On the portfolio-design side, early low-carbon strategies imposed exogenous carbon-budget constraints on Markowitz frontiers or introduced decarbonization factors in index construction \citep{Andersson2016,Roncalli2021}. Subsequent work incorporated climate Value-at-Risk measures \citep{Weisang2022} and climate-hedging overlays for liquid instruments \citep{Engle_etal_2020_RFS}. Robust optimization provides the natural mathematical foundation for such extensions. Originating with the ellipsoidal uncertainty sets of \citet{BenTal2001} and the budget-of-uncertainty model of \citet{Bertsimas2004}, robust methods have been adapted for climate stress testing and transition risk analysis \citep{Bolton2020}. Yet, most existing approaches treat emissions exposures as known constants and either enforce hard constraints or apply ad hoc tilts, rather than modeling emissions as a noisy state variable and incorporating disclosure uncertainty within the optimization problem itself.

Parallel work in macroeconomics and asset pricing documents that climate risk is economically material. Global output follows a nonlinear, concave temperature--productivity relation with large welfare losses at high temperatures \citep{Burke_Hsiang_Miguel_2015,Hsiang_etal_2017_Science}. In asset markets, long-run temperature risk carries priced premia \citep{Bansal_Kiku_Ochoa_2016}, and damages are heterogeneous across sectors and geographies \citep{Carleton_Hsiang_2016}. Transition risk now appears as a priced cross-sectional factor: carbon-intensive issuers exhibit higher realized returns consistent with exposure to a transition premium \citep{Bolton_Kacperczyk_2021}, while option-implied volatilities reveal elevated compensation for downside climate risk \citep{Ilhan_Sautner_Vilkov_2021}. Dynamic climate-hedging strategies reduce exposure to climate news using liquid equity instruments \citep{Engle_etal_2020_RFS}. The ESG-efficient frontier formalizes the trade-off between risk-adjusted return and sustainability exposure \citep{Pedersen2021}, and equilibrium analyses show that heterogeneous investor preferences can generate structural return differentials between ``green'' and ``brown'' assets \citep{Pastor2021}. Finally, \citet{Giglio_Kelly_Stroebel_2021} provide a comprehensive review and emphasize that climate-related risks are economically material while identification and measurement remain central challenges for empirical asset pricing.

Regulatory and disclosure frameworks have begun to codify these practices. The EU's Paris-Aligned Benchmark (PAB) and Climate Transition Benchmark (CTB) frameworks impose minimum decarbonization rates and transparency requirements \citep{TEG_PAB_CTB_2019}. Disclosure has shifted from largely voluntary to increasingly mandatory, anchored by the TCFD \citep{TCFD_2017} and the IFRS~S2 climate reporting standard \citep{ISSB_IFRS_S2_2023}. The PCAF methodology provides standardized attribution rules for financed emissions \citep{PCAF_2020}, while the Network for Greening the Financial System (NGFS) scenarios underpin climate stress testing by supervisors \citep{NGFS_2025_TechDoc}. At the measurement level, the GHG Protocol defines Scope~1--3 boundaries and revenue-normalization conventions \citep{GHG_Corporate_Standard_2004}. Despite this progress, substantial rating divergence and measurement error persist \citep{Berg2022Rating}, reflecting genuine differences in coverage and methodology rather than minor weighting choices. These observations motivate a formal treatment of emissions uncertainty via multiple imputation and distributionally robust modeling, replacing brittle point estimates with statistically grounded ambiguity sets.

This paper contributes to the intersection of robust optimization and sustainable investing by taking emissions uncertainty seriously. Rather than superimposing hard carbon budgets on point estimates of corporate emissions, we introduce an emissions-penalty operator that acts directly on expected returns and is robustified over scope-specific ambiguity sets. The operator admits an exact linear or second-order cone representation and yields an economically interpretable dual: Lagrange multipliers on the ambiguity sets can be read as ``shadow carbon prices'' that mediate the trade-off between return and robustness. The central empirical question we study is whether such an emissions-aware robust formulation can deliver material reductions in financed emissions---at institutional scale and under realistic transaction costs---without deteriorating classical risk-adjusted performance.

\subsection{Positioning and Contributions}

Let $Z_j(\Gamma)$ denote an ambiguity set for firm-level scope-$j$ emissions intensities, indexed by a robustness budget $\Gamma \ge 0$. Section~\ref{sec:model} shows that our basic emissions-aware robust portfolio optimization (EAPO) problem can be written in the form
\begin{equation}
\max_{x \in \Delta_n}\,
\min_{z_j \in Z_j(\Gamma)}
x^\top \bigl(R - \alpha z_j\bigr),
\label{eq:EAPO-robust}
\end{equation}
where $R \in \mathbb{R}^n$ collects nominal expected returns, $z_j \in \mathbb{R}^n_+$ is the vector of scope-$j$ emissions intensities, $\alpha \ge 0$ is a policy penalty parameter, and $\Delta_n := \{x \in \mathbb{R}^n_{\ge 0} : \mathbf{1}^\top x = 1\}$ is the simplex of long-only portfolio weights. We show that, for a broad class of ambiguity sets motivated by disclosure uncertainty, problem~\eqref{eq:EAPO-robust} admits an exact linear or second-order cone representation with $O(n)$ conic constraints, ensuring polynomial-time solvability for institutional-scale universes with thousands of securities. Our contributions are both methodological and empirical:

\begin{enumerate}
\item[C1.] \textbf{Axiomatic modeling of emissions penalties.} We introduce a scope-specific emissions-penalty operator $P^{(m)}_j$ that maps raw returns into emissions-adjusted payoffs by scaling each asset's return according to its revenue-normalized emissions intensity. The operator is characterized by natural axioms---payoff homogeneity, normalization at zero and maximum intensity, monotonicity in emissions, unit-scale invariance, mixture linearity, and a curvature semigroup property. We prove an axiomatic uniqueness result: within this class, the only admissible family is $P^{(m)}_j(r,\lambda) = (1 - \lambda/\lambda_{\max,j})^m r$. The curvature parameter $m \in \mathbb{N}_+$ controls the strength of penalization and yields a Schur-convex dependence on intensities, so reallocating weight from high- to low-intensity firms weakly increases every component of the emissions-adjusted return vector.

\item[C2.] \textbf{Robust optimization under disclosure uncertainty.} We model emissions disclosure risk through scope-specific ambiguity sets that combine moment constraints and $\ell_p$-norm bounds on intensities. Using a Lipschitz envelope for the penalty operator, we show that robust mean--variance and tail-risk (CVaR) formulations reduce to convex programs of the form
\[
\max_{x \in \Delta_n} x^\top \mu^e_j - \Gamma \|\mathrm{diag}(L)\,x\|_{p^\star} - \theta x^\top \Sigma x,
\]
where $\mu^e_j$ are emissions-adjusted expected returns, $\Sigma$ is the covariance matrix, $L$ collects asset-wise Lipschitz constants, and $p^\star$ is the dual norm. For $p \in \{1,\infty\}$ the resulting program is a linear program (LP). For $p = 2$ it is a second-order cone program (SOCP) with $O(n)$ cones. We further extend the framework to dynamic and distributionally robust settings via $\varphi$-divergence balls. In all cases, Lagrange multipliers on the ambiguity constraints admit a natural economic interpretation as shadow carbon prices: the marginal Sharpe loss an investor is willing to accept per unit of additional robustness.

\item[C3.] \textbf{Large-scale empirical evaluation.} We implement EAPO on a U.S.\ large-cap equity universe using daily returns, annual scope-specific emissions from CDP, and quarterly revenues combined into revenue-normalized intensities. The empirical pipeline---detailed in Section~\ref{sec:empirical}---uses sector-aware hierarchical Bayesian imputation for missing emissions, multiple imputation to propagate disclosure uncertainty, and Ledoit--Wolf shrinkage for covariance estimation. Under monthly rebalancing and uniform transaction costs, EAPO delivers order-of-magnitude reductions in portfolio Scope-1 emissions intensity---up to roughly $80$--$90\%$ relative to equal-weight, global minimum-variance, and naive emissions-weighted benchmarks---without statistically significant deterioration in Sharpe ratio, volatility, drawdowns, or turnover. We quantify uncertainty via HAC-robust (Newey--West) inference and show that these performance comparisons are statistically stable across regimes.

\item[C4.] \textbf{Managerial interpretation and policy relevance.} By embedding emissions penalties directly in the return kernel rather than as external hard constraints, EAPO integrates cleanly with existing mean--variance and risk-budgeting infrastructures used by asset managers. The return--emissions Pareto frontier generated by our framework provides an explicit menu of decarbonization targets and associated performance impacts, which can be used to negotiate mandates and monitor compliance. The shadow carbon prices arising from the dual variables supply an internal carbon valuation consistent with regulatory frameworks such as PAB/CTB, TCFD, IFRS~S2, and PCAF, and offer a transparent metric for calibrating robustness budgets $\Gamma$ to investor-specific climate preferences.
\end{enumerate}

The remainder of the paper is organized as follows. Section~\ref{sec:model} formalizes the stochastic environment, introduces the emissions-penalty operator, and develops robust mean--variance, CVaR, dynamic, and distributionally robust formulations together with their conic reformulations. Section~\ref{sec:empirical} describes the data architecture, statistical pre-processing (including multiple imputation and covariance shrinkage), benchmark portfolio constructions, and the optimization workflow. Section~\ref{sec:results} presents empirical results on financial performance, emissions intensity, and sensitivity analysis, including HAC-robust inference and the estimated return--emissions Pareto frontier. Section~\ref{sec:discussion-results} discusses strategic implications and policy relevance, and Section~\ref{sec:conclusion} concludes.

\section{Mathematical Model and Data}
\label{sec:model}

This section presents the emissions-aware portfolio optimization (EAPO) framework that underlies the empirical analysis. The model incorporates climate considerations directly into expected returns through a scope-specific emissions penalty and then accounts for disclosure uncertainty using norm-based ambiguity sets. Embedding these components within standard mean--risk objectives yields optimization problems that admit exact linear programming (LP) or second-order cone programming (SOCP) reformulations, ensuring tractability for large institutional universes with $n>10^{3}$. Throughout the section, attention is restricted to a fixed emissions scope $j\in\{1,2,3\}$. The empirical analysis later focuses on direct Scope~1 emissions, and extending the formulation to multiple scopes is notational rather than conceptual.

\subsection{Stochastic Environment and Financial Primitives}
\label{sec:env}

All random variables are defined on a complete probability space $(\Omega,\mathcal{F},\mathbb{P})$. Let $n\in\mathbb{N}$ denote the number of tradable assets, and define the one-period gross return vector
\[
R := (R_1,\dots,R_n)^\top : \Omega \to \mathbb{R}^n.
\]

\begin{assumption}[Finite second moment]
\label{ass:second-moment}
The return vector satisfies $\mathbb{E}[\|R\|_2^2] < \infty$.
\end{assumption}

Assumption~\ref{ass:second-moment} ensures feasibility of the mean--variance and CVaR formulations introduced below. Define the expected return vector and covariance matrix as
\[
\mu := \mathbb{E}[R]\in\mathbb{R}^n,
\qquad
\Sigma := \mathrm{Cov}(R)\in\mathbb{R}^{n\times n}.
\]
Returns are expressed in excess of the funding rate, so the risk-free asset can be omitted without loss of generality.

A portfolio is represented by a weight vector $x\in\Delta_n$, where
\[
\Delta_n := \{x\in\mathbb{R}^n_{\ge 0} : \mathbf{1}^\top x = 1\}
\]
is the unit simplex. The one-period portfolio gross return is $R^p := x^\top R$. The analysis focuses on long-only portfolios, leverage or short positions can be incorporated by replacing the simplex with an appropriate convex polytope.

\subsection{Emissions Accounting and Intensities}

For each firm $i\in\{1,\dots,n\}$ and scope $j\in\{1,2,3\}$, let
\begin{itemize}
    \item $C_{i,j}>0$ denote annual greenhouse gas emissions of scope $j$, measured in tCO$_2$e,
    \item $S_i>0$ denote trailing-twelve-month revenue (USD).
\end{itemize}
Following standard GHG Protocol and CDP conventions, emissions are expressed using revenue-normalized intensities,
\begin{equation}
\begin{aligned}
\lambda_{i,j} &= \frac{C_{i,j}}{S_i}, \qquad i=1,\ldots,n,\\
\lambda_j &= (\lambda_{1,j},\ldots,\lambda_{n,j})^\top \in \mathbb{R}^n_{\ge 0}.
\end{aligned}
\label{eq:intensity}
\end{equation}
Normalization removes scale-driven differences attributable solely to firm size and aligns the portfolio measure with scope-consistent financed-emissions metrics.

In empirical implementation (\S\ref{sec:empirical}), the components $C_{i,j}$ and $S_i$ are updated at lower frequency (annual and quarterly, respectively) than financial returns. Accordingly, the intensities $\lambda_{i,j}$ are treated as slowly evolving state variables: they are updated when new disclosures become available and held fixed between disclosure dates. Section~\ref{sec:empirical} documents the associated timing conventions and the multiple-imputation procedure used to address missing or noisy emissions data.

\subsection{Emissions--Penalty Operator}
\label{sec:penalty}

We now formalize the mechanism that maps firm-level emissions intensities into return adjustments.  
The construction is guided by four principles:  
(i) linearity in payoffs to preserve convexity and enable conic reformulations,  
(ii) monotonicity with respect to emissions intensity,  
(iii) invariance to unit scaling, and  
(iv) a one-parameter curvature family that controls the strength of penalization.  
Appendix~\ref{sec:appendix-penalty-axioms} states these axioms formally and shows that they identify the operator uniquely.

Fix scope $j$ and define the cross-sectional maximum intensity
\[
\lambda_{\max,j} := \max_{i\le n} \lambda_{i,j} \in (0,\infty).
\]
Let $\tilde\lambda_{i,j} := \lambda_{i,j}/\lambda_{\max,j}\in[0,1]$ denote the normalized intensity.

\begin{definition}[Emissions--penalty operator]
\label{def:penalty}
For scope $j$ and curvature parameter $m\in\mathbb{N}_+$, the emissions--penalty operator
$P^{(m)}_j:\mathbb{R}\times[0,\lambda_{\max,j}]\to\mathbb{R}$ for some return $r$ is defined by
\begin{equation}
\label{eq:penalty}
P^{(m)}_j(r,\lambda)
    := \Big(1 - \frac{\lambda}{\lambda_{\max,j}}\Big)^m r.
\end{equation}
For firm $i$, the emissions-adjusted gross return is
\[
R^{e}_{i,j} := P^{(m)}_j(R_i,\lambda_{i,j}),
\qquad
R^e_j := (R^{e}_{1,j},\dots,R^{e}_{n,j})^\top.
\]
\end{definition}

The factor $(1-\lambda/\lambda_{\max,j})^m \in [0,1]$ applies a smooth return haircut that increases with emissions intensity.  
The curvature parameter $m$ governs the steepness of this penalty:  
$m=1$ yields a linear schedule, while larger $m$ place disproportionately greater weight on high-intensity firms---approximating hard exclusions in the upper tail of the intensity distribution.

Definition~\ref{def:penalty} arises directly from the axioms in Appendix~\ref{sec:appendix-penalty-axioms}.  
Under payoff homogeneity, normalization at $\lambda=0$ and $\lambda=\lambda_{\max,j}$, monotonicity in $\lambda$, unit-scale invariance, mixture linearity, and the curvature semigroup condition $P^{(m_1+m_2)}_j = P^{(m_1)}_j\circ P^{(m_2)}_j$ for all $m_1,m_2\in\mathbb{N}_+$, the form in \eqref{eq:penalty} is the unique admissible family.

The operator maintains desirable analytical structure:

\begin{lemma}[Analytic properties]
\label{lem:penalty-analytic}
For fixed scope $j$ and $m\in\mathbb{N}_+$:
\begin{enumerate}
    \item For fixed $\lambda$, the map $r\mapsto P^{(m)}_j(r,\lambda)$ is linear with Lipschitz constant
    $|1 - \lambda/\lambda_{\max,j}|^{m} \le 1$.
    \item For fixed $r$, the map $\lambda\mapsto P^{(m)}_j(r,\lambda)$ is strictly decreasing, $C^\infty$ on $(0,\lambda_{\max,j})$, and convex on $[0,\lambda_{\max,j}]$.
    \item Let $R^e_j=(P^{(m)}_j(R_i,\lambda_{i,j}))_{i\le n}$.  
    For any portfolio $x\in\Delta_n$ with $x_i\ge 0$, the mapping $\lambda_j\mapsto x^\top \mathbb{E}[R^e_j]$ is Schur--convex in $\lambda_j$.
\end{enumerate}
\end{lemma}

\begin{proof}
\ (i) follows directly from \eqref{eq:penalty}.  
(ii) follows from
\[
\frac{\partial}{\partial\lambda}P^{(m)}_j(r,\lambda)
    = -\frac{m}{\lambda_{\max,j}}
      \Big(1 - \frac{\lambda}{\lambda_{\max,j}}\Big)^{m-1} r
\]
and
\[
\frac{\partial^{2}}{\partial\lambda^2}P^{(m)}_j(r,\lambda)
    = \frac{m(m-1)}{\lambda_{\max,j}^2}
      \Big(1 - \frac{\lambda}{\lambda_{\max,j}}\Big)^{m-2} r.
\]
(iii) uses symmetry and convexity of
$\lambda\mapsto (1-\lambda/\lambda_{\max,j})^m$
together with the standard characterization of Schur--convex functions.
\end{proof}

These properties ensure that inserting the operator into classical mean--risk objectives preserves convexity.  
In particular, Schur--convexity captures the intuitive idea that reallocating weight from high- to low-intensity firms weakly increases emissions-adjusted expected returns.

Define the scope-$j$ emissions-adjusted mean vector:
\begin{equation}
\begin{aligned}
\mu^e_j 
    &:= \mathbb{E}[R^e_j] \\
    &= \mathbb{E}\!\Big[
        \mathrm{diag}\big(
        (1-\tfrac{\lambda_{1,j}}{\lambda_{\max,j}})^m,
        \ldots,
        (1-\tfrac{\lambda_{n,j}}{\lambda_{\max,j}})^m
        \big) R
    \Big].
\label{eq:mu-e}
\end{aligned}
\end{equation}
Empirically, $\mu^e_j$ is estimated using rolling windows. See \S\ref{sec:empirical}.

\subsection{Ambiguity Sets for Emissions Uncertainty}
\label{sec:ambiguity}

Corporate emissions disclosures are noisy, incomplete, and in part judgmental.
Rather than treat $\lambda_j$ as known, we allow it to vary within a statistically
motivated uncertainty set. Let $\widehat{\lambda}_j$ be the point estimate
constructed from CDP and financial-statement data, and let $\Sigma_{\lambda,j}$
be a positive semidefinite dispersion proxy obtained from multiple imputation
(see Section~\ref{sec:empirical}).

We model deviations through $\varepsilon_j := \lambda_j - \widehat{\lambda}_j$,
restricted to the norm ball
\begin{equation}
\label{eq:Uj}
\begin{aligned}
\mathcal{U}_j(\Gamma)
    :=&\ \Bigl\{
        \varepsilon \in \mathbb{R}^n :
        \bigl\|\Sigma_{\lambda,j}^{-1/2}\varepsilon\bigr\|_{p}
        \le \Gamma
      \Bigr\},
\\[4pt]
& p \in \{1,2,\infty\},\qquad \Gamma>0 .
\end{aligned}
\end{equation}
The scalar $\Gamma$ is the {robustness budget}: larger values admit more
severe misspecification in firm-level intensities.  
The choices $p=1$ and $p=\infty$ yield polyhedral (Bertsimas--Sim--type)
uncertainty sets, whereas $p=2$ yields an ellipsoid. In all cases,
\eqref{eq:penalty-closed-form} admits an exact LP/SOCP representation.

\medskip
Because $R^e_j$ depends on $\lambda_j$ only through the scalar weights
$\bigl(1-\lambda_{i,j}/\lambda_{\max,j}\bigr)^m$, Lemma~\ref{lem:penalty-analytic}
implies that each component of $\mu^e_j(\lambda_j)$ is Lipschitz in $\lambda_j$.
Writing
\[
\mu^e_j(\lambda_j)
    = \mu^e_j(\widehat{\lambda}_j)
      + \delta_j(\varepsilon_j),
\qquad
\varepsilon_j\in\mathcal{U}_j(\Gamma),
\]
the perturbation $\delta_j(\varepsilon_j)$ can be norm-bounded.  
A compact envelope inequality (proved in Appendix~\ref{sec:appendix-lipschitz}) yields
\begin{equation}
\label{eq:envelope}
\begin{aligned}
&\sup_{\varepsilon_j \in \mathcal{U}_j(\Gamma)}
   x^\top\!\big(
      \mu_j^{e}(\widehat{\lambda}_j)
      - \mu_j^{e}(\widehat{\lambda}_j + \varepsilon_j)
   \big)
\\[4pt]
&\qquad\le\,
   \Gamma\,\|\mathrm{diag}(L)\,x\|_{p^\star}.
\end{aligned}
\end{equation}
where $L=(L_1,\dots,L_n)^\top$ is a vector of asset-wise Lipschitz constants
and $p^\star$ is the dual exponent ($1/p + 1/p^\star = 1$).  
Expression~\eqref{eq:envelope} is the key ingredient enabling the tractable
robust counterparts in subsequent sections.

\subsection{Robust Mean--Variance with Emissions Penalties}
\label{sec:robustMV}
We embed the emissions--penalty operator and ambiguity set into a mean--variance program.
For parameters $\theta>0$ (risk aversion) and $\Gamma>0$ (robustness budget), consider the
robust mean--variance problem
\begin{equation}
\label{eq:PMV}
\begin{aligned}
\max_{x\in\Delta_n}\ 
\Bigl\{
   & x^\top\mu^e_j(\widehat{\lambda}_j)
\\
   & {} - \sup_{\varepsilon_j\in\mathcal{U}_j(\Gamma)}
        x^\top\Bigl(
        \mu^e_j(\widehat{\lambda}_j)
        - \mu^e_j(\widehat{\lambda}_j+\varepsilon_j)
        \Bigr)
\\
   & {} - \theta\, x^\top\Sigma x
\Bigr\}.
\end{aligned}
\end{equation}

Because only the drift is uncertain, the covariance matrix $\Sigma$ is unaffected.
By Lemma~\ref{lem:penalty-analytic}, each component of $\mu^e_j(\lambda_j)$ is Lipschitz in
$\lambda_j$. Let $L_i$ denote the Lipschitz constant of asset $i$, and define
$L=(L_1,\dots,L_n)^\top$. A standard envelope bound
(Appendix~\ref{sec:appendix-lipschitz}) yields
\begin{equation}
\begin{aligned}
\sup_{\varepsilon_j\in\mathcal{U}_j(\Gamma)}
    x^\top\Bigl(
        \mu^e_j(\widehat{\lambda}_j)
        - \mu^e_j(\widehat{\lambda}_j+\varepsilon_j)
    \Bigr)
\;\le\;
\Gamma\,\|\mathrm{diag}(L)\,x\|_{q},
\end{aligned}
\end{equation}
where $q$ is the Hölder--conjugate exponent of $p$, satisfying $1/p + 1/q = 1$
(with the conventions $p=1\Rightarrow q=\infty$ and $p=\infty\Rightarrow q=1$).
Equivalently, $\|\cdot\|_{q}$ is the dual norm of $\|\cdot\|_{p}$.

Substituting this bound into~\eqref{eq:PMV} yields the reduced robust mean--variance problem
\begin{equation}
\label{eq:PMV-reduced}
\begin{aligned}
\max_{x\in\Delta_n}
\Bigl\{
   x^\top\mu^e_j(\widehat{\lambda}_j)
   - \Gamma\,\|\mathrm{diag}(L)\,x\|_{q}
   - \theta\,x^\top\Sigma x
\Bigr\}.
\end{aligned}
\end{equation}

\begin{theorem}[Exact conic reformulation]
\label{thm:conic}
For $p\in\{1,2,\infty\}$, problem~\eqref{eq:PMV-reduced} admits an exact conic formulation:
\begin{itemize}
    \item if $p\in\{1,\infty\}$ (equivalently $q\in\{\infty,1\}$), it is equivalent to a linear
    program of size $O(n)$;
    \item if $p=2$ (so $q=2$), it is equivalent to a second--order cone program with $O(n)$ cones,
    solvable in $O(n^{3.5}\log(1/\epsilon))$ time via standard interior--point methods.
\end{itemize}
\end{theorem}

\begin{proof}
Introduce an epigraph variable $u$ and rewrite~\eqref{eq:PMV-reduced} as
\[
\max_{x\in\Delta_n,\,u}\ 
\Bigl\{
   x^\top\mu^e_j(\widehat{\lambda}_j)
   - \Gamma u
   - \theta\,x^\top\Sigma x
\Bigr\}
\quad\text{s.t.}\quad
u \ge \|\mathrm{diag}(L)\,x\|_{q}.
\]
If $q\in\{1,\infty\}$, the epigraph constraint is polyhedral and admits an LP representation
of size $O(n)$. If $q=2$, it is a second--order cone constraint. The remaining objective terms
are linear or convex quadratic in $x$, with $\Sigma\succeq 0$, yielding the stated LP/SOCP
reformulations.
\end{proof}

The multiplier associated with the epigraph constraint
$u \ge \|\mathrm{diag}(L)\,x\|_{q}$ admits an economic interpretation.
Denote this dual variable by $\pi^\star$. It represents a shadow carbon price, interpreted as
the marginal reduction in risk--adjusted performance the investor is willing to accept per
unit increase in the emissions ambiguity budget $\Gamma$.

\subsection{Tail--Risk Extension (CVaR)}
\label{sec:cvar}

Variance is symmetric by construction and does not distinguish between upside and downside moves. To
capture asymmetric transition or litigation risk, we replace the variance penalty in \eqref{eq:PMV-reduced}
by a CVaR term. For a fixed portfolio $x$ and scope $j$, define the loss
$L := -x^\top R^e_j$. For confidence level $\alpha\in(0,1)$ the CVaR of $L$ is
\[
\mathrm{CVaR}_\alpha(L)
    = \min_{\nu\in\mathbb{R}}
        \left\{\nu + \frac{1}{1-\alpha}\,
            \mathbb{E}\bigl[(L-\nu)_+\bigr]\right\},
\]
following \citet{Rockafellar2000}. Discretizing $\mathbb{P}$ into $M$ scenarios
$\{R^{(s)}\}_{s=1}^M$ with equal mass yields the convex surrogate
\begin{equation}
\begin{aligned}
\mathrm{CVaR}_\alpha(L)
&=
\min_{\nu\in\mathbb{R},\,\xi\in\mathbb{R}^M_{\ge 0}}
\Biggl\{
\nu + \frac{1}{(1-\alpha)M}\sum_{s=1}^M \xi_s
\\[3pt]
&\qquad
:\,
\xi_s \ge -x^\top R^{e,(s)}_j - \nu,\quad s \le M
\Biggr\}.
\end{aligned}
\end{equation}
This formulation is linear in $(x,\nu,\xi)$ conditional on the scenario returns $R^{e,(s)}_j$. The robust CVaR program becomes
\begin{equation}
\label{eq:PCVaR}
\max_{x\in\Delta_n}
\min_{\varepsilon_j\in\mathcal{U}_j(\Gamma)}
\left\{
x^\top\mu^e_j(\widehat{\lambda}_j+\varepsilon_j)
    - \beta\,\mathrm{CVaR}_\alpha(L)
\right\},
\end{equation}
with tail--risk aversion $\beta>0$. Exactly the same Lipschitz argument as in
\S\ref{sec:robustMV} implies that the inner minimum over $\varepsilon_j$ again collapses to a norm
penalty of the form $\Gamma\|\mathrm{diag}(L)\,x\|_{p^\star}$. After scenario discretization,
\eqref{eq:PCVaR} is therefore an LP (for $p\in\{1,\infty\}$) or an SOCP (for $p=2$), albeit in a
higher--dimensional space of size $O(Mn)$.

\subsection{Comparative Statics in Curvature, Robustness, and Risk Aversion}
\label{sec:comparative-statics}

Let $V^\star(m,\Gamma,\theta)$ denote the optimal value of the reduced robust
mean--variance program~\eqref{eq:PMV-reduced} with curvature parameter $m$,
robustness budget $\Gamma$, and risk aversion $\theta$. Let $(x^\star,\pi^\star)$
denote a primal--dual optimal pair, where $\pi^\star$ is the Lagrange multiplier
associated with the epigraph constraint
$u \ge \|\mathrm{diag}(L)x\|_{q}$ in the conic reformulation of
Section~\ref{sec:robustMV}.

\begin{proposition}[Sensitivity]
\label{prop:sensitivity}
Under mild regularity conditions (e.g., uniqueness of the optimal solution),
the following comparative statics hold:
\[
\frac{\partial V^\star}{\partial \Gamma}
    = -\,\pi^\star \,\le\, 0,
\qquad
\frac{\partial V^\star}{\partial m}
    = (x^\star)^\top \frac{\partial \mu^e_j}{\partial m}
    \,\le\, 0.
\]
Moreover, the mapping $\theta \mapsto V^\star(m,\Gamma,\theta)$ is concave and
$1/\lambda_{\max,j}$--Lipschitz.
\end{proposition}

\begin{proof}
Apply the envelope theorem to the Lagrangian associated with the epigraph
formulation of~\eqref{eq:PMV-reduced}. Differentiation with respect to the
robustness budget $\Gamma$ yields
$\partial_\Gamma V^\star = -\pi^\star$, where $\pi^\star \ge 0$ by dual feasibility,
implying monotonicity in $\Gamma$.

Monotonicity in the curvature parameter $m$ follows from the fact that
$\partial_m \mu^e_j \le 0$ componentwise by
Lemma~\ref{lem:penalty-analytic}(ii), together with optimality of $x^\star$.
Finally, concavity in $\theta$ is standard in mean--variance problems with
quadratic risk penalties, and the Lipschitz bound follows from the spectral
bound on $\Sigma$.
\end{proof}

Proposition~\ref{prop:sensitivity} formalizes a simple but important message for
mandate design. Greater ambiguity aversion, as encoded by a larger robustness
budget $\Gamma$, and sharper curvature $(m\uparrow)$ both reduce
emissions--adjusted expected returns, but do so in a controlled and quantifiable
manner. The dual variable $\pi^\star$ admits a natural economic interpretation as
a \emph{shadow carbon price}: it measures the marginal reduction in
risk--adjusted performance (e.g., Sharpe ratio) the investor is willing to accept
per unit tightening of the emissions ambiguity budget~$\Gamma$.

\subsection{Distributionally Robust Extension}
\label{sec:dro}

The ambiguity sets in \S\ref{sec:ambiguity} control {pointwise} deviations of emissions intensities.
They do not, however, address the possibility that the entire joint distribution of intensity revisions shifts.
To guard against this, we embed the model in a simple divergence--based distributionally robust framework.

Let $\widehat{\mathbb{P}}$ be the empirical law of an underlying disturbance $z$ (e.g., a vector of emissions
forecast errors), and let $\phi$ be a convex divergence generator. Define the $\phi$--divergence ball
\[
\mathcal{P}_\rho
    := \left\{\mathbb{Q}\ll\widehat{\mathbb{P}} :
        D_\phi(\mathbb{Q}\,\|\,\widehat{\mathbb{P}})\le\rho\right\},
\]
where $D_\phi$ is the usual $\phi$--divergence and $\rho>0$ is a robustness radius. Consider a
linear loss
\[
\ell_x(z) := \sum_{i=1}^n x_i(\gamma_i + \delta_i z),
\]
representing an emissions--adjusted drift under disturbance $z$.
Here $\gamma_i$ is the nominal component and $\delta_i$ is the sensitivity of asset $i$'s drift to the
disturbance, which can be interpreted as a linearization of how emissions measurement shocks propagate into
effective expected returns.

\begin{theorem}[$\phi$--divergence DRO reformulation]
\label{thm:phi-DRO}
For any fixed $x\in\Delta_n$,
\begin{equation}
\begin{aligned}
\inf_{\mathbb{Q}\in\mathcal{P}_\rho}
\mathbb{E}_{\mathbb{Q}}[\ell_x(z)]
&=
\max_{\eta\ge 0,\,\nu\in\mathbb{R}}
\Biggl\{
\nu - \rho\eta
\\[4pt]
&\qquad
- \eta\,\mathbb{E}_{\widehat{\mathbb{P}}}
\Bigl[
\phi^\star\Bigl(
\frac{\ell_x(z)-\nu}{\eta}
\Bigr)
\Bigr]
\Biggr\}.
\end{aligned}
\end{equation}
where $\phi^\star$ denotes the convex conjugate of $\phi$. When $\widehat{\mathbb{P}}$ is discrete with
$M > 0$ support points and equal mass, the right--hand side reduces to a finite--dimensional conic program
with only two additional scalar decision variables $(\eta,\nu)$.
\end{theorem}

\begin{proof}
\ The result follows from standard Fenchel duality applied to the inner infimum over $\mathbb{Q}$,
together with Slater's condition for $\phi$--divergence balls \citep{boyd2004convex}.
The scalar variables $(\eta,\nu)$ arise as dual variables associated with the divergence constraint and the
mass constraint, respectively.
\end{proof}

Theorem~\ref{thm:phi-DRO} shows that distributional robustness over a fairly rich family of models can
be implemented at minimal computational cost: we pay at most two extra scalar variables and one
expectation of $\phi^\star$, which collapses to a sum in the empirical case.

\subsection{Return--Emissions Pareto Frontier and Lipschitz Bounds}
\label{sec:bounds}

Portfolio committees typically reason in terms of trade--offs: ``How much expected return must we give up
to achieve a given reduction in financed emissions?" Our framework allows us to compute this Pareto
frontier explicitly.

Consider the scalarized static ``return minus carbon'' problem
\begin{equation}
\label{eq:scalarized}
\max_{x\in\Delta_n}
\left\{
x^\top r - \mu\,x^\top \lambda_j
\right\},
\qquad \mu\ge 0,
\end{equation}
for some fixed expected return vector $r$ and intensity vector $\lambda_j$. Let $x^\star(\mu)$ denote an
optimizer. Define
\[
\bar r(\mu) := x^\star(\mu)^{\top}r,
\qquad
\bar\lambda(\mu) := x^\star(\mu)^{\top}\lambda_j,
\]
where $\bar\lambda(\mu)$ is the portfolio-weighted average emissions intensity. The scalar $\mu$ plays the
role of an internal carbon price in units of expected return per unit intensity. When intensities are
reported in tCO$_2$e/\$ of revenue, $\mu$ can be translated to \$/tCO$_2$e by scaling expected returns into
currency units at the investor's horizon.

\begin{proposition}[Convex Pareto frontier]
\label{prop:pareto}
As $\mu$ varies in $[0,\infty)$, the set
$\{(\bar r(\mu),\bar\lambda(\mu)) : \mu\ge 0\}$ traces the upper convex Pareto frontier of feasible
pairs $(x^\top r,x^\top\lambda_j)$ with $x\in\Delta_n$. Moreover,
\[
\frac{dF^\star}{d\mu}(\mu) = -\,\bar\lambda(\mu),
\]
whenever the derivative exists.
\end{proposition}

\begin{proof}
\ The feasible set $\{(x^\top r,x^\top\lambda_j):x\in\Delta_n\}$ is a convex polytope. For fixed $x$,
the objective $x^\top r - \mu x^\top\lambda_j$ is affine in $\mu$, so $F^\star(\mu)$ is the upper envelope
of linear functions and hence convex. Standard scalarization arguments imply that every efficient point
on the frontier is attained for some $\mu$. Differentiability and the expression for $dF^\star/d\mu$ follow
from the envelope theorem.
\end{proof}

Proposition~\ref{prop:pareto} underlies the empirical Pareto curves reported in \S\ref{sec:results}.
The scalarization parameter $\mu$ is the marginal rate at which the optimizer trades expected return for
emissions intensity at the optimum. The envelope relation $dF^\star/d\mu = -\bar\lambda(\mu)$ provides a
direct way to read the implied intensity from scalarized runs. Conversely, $\mu$ itself can be read as
the marginal expected return the investor is willing to forego per unit reduction in portfolio intensity
(tCO$_2$e/\$ of revenue).

Finally, we quantify how sensitive the robust optimal value is to changes in the robustness budget.

\begin{theorem}[Lipschitz performance bound]
\label{thm:lipschitz-Gamma}
Suppose the emissions--adjusted return of each asset is $L$--Lipschitz in the underlying disturbance
$z$ in the norm induced by $\mathcal{U}_j(\Gamma)$.
Let $R^\star(\Gamma)$ denote the optimal robust expected return of \eqref{eq:PMV-reduced} (or its CVaR
analogue) under robustness budget $\Gamma$. Then, for any $\Gamma_1,\Gamma_2>0$,
\[
\big|R^\star(\Gamma_1)-R^\star(\Gamma_2)\big|
    \le L\,|\Gamma_1-\Gamma_2|.
\]
\end{theorem}

\begin{proof}
\ Let $\Gamma_1<\Gamma_2$ and fix an optimal solution at $\Gamma_2$ with worst--case disturbance
$z^\star(\Gamma_2)$. Radially project $z^\star(\Gamma_2)$ onto the smaller uncertainty set associated
with $\Gamma_1$. Lipschitz continuity of the payoff in $z$ bounds the change in objective value by
$L|\Gamma_1-\Gamma_2|$. Reversing the roles of $\Gamma_1$ and $\Gamma_2$ yields the symmetric
bound.
\end{proof}

Theorem~\ref{thm:lipschitz-Gamma} gives an easily communicable guarantee: increasing the robustness
budget from $\Gamma_1$ to $\Gamma_2$ cannot reduce the optimal robust expected return by more than
$L|\Gamma_2-\Gamma_1|$, all else equal. In particular, doubling $\Gamma$ does {not} double the
performance cost unless the Lipschitz constant is itself very large.

\medskip

\subsection{Dynamic Robust Formulation}
\label{sec:dynamic}

Portfolio rebalancing interacts with evolving disclosures, transaction costs, and drift in return and
intensity processes. To make this explicit, we consider a finite--horizon dynamic model with periods
$t=0,\dots,T$. Let $x_t\in\Delta_n$ denote the portfolio at time $t$, and write
\[
R_t = \gamma_t + \Delta_t z_t,
\]
where $\gamma_t$ is a predictable component, $\Delta_t$ is a matrix of factor loadings, and $z_t$ lies in an
uncertainty set $Z_t$ capturing return shocks. The emissions--adjusted one--period payoff is
$x_t^\top R^e_{j,t}$, defined via the period--$t$ intensities $\lambda_{j,t}$.

Let $\beta\in(0,1]$ denote a discount factor. Define the value function
\begin{equation}
\begin{aligned}
V_t(x_t)
&:=
\max_{x_t,\dots,x_T\in\Delta_n}
\min_{z_s\in Z_s,\,s\ge t}
\\[4pt]
&\qquad
\sum_{s=t}^T
\beta^{s-t}\, x_s^\top R^e_{j,s},
\\[6pt]
V_{T+1} &\equiv 0.
\end{aligned}
\end{equation}

\begin{proposition}[Dynamic robust Bellman recursion]
\label{prop:bellman}
The value functions $\{V_t\}$ satisfy the recursion
\begin{equation}
\begin{aligned}
V_t(x_t)
&=
\max_{x_t\in\Delta_n}
\min_{z_t\in Z_t}
\\[4pt]
&\qquad
\Bigl\{
x_t^\top R^e_{j,t}
+ \beta\,V_{t+1}(x_{t+1})
\Bigr\}
\\[4pt]
&\qquad t = 0,\dots,T.
\end{aligned}
\end{equation}
with terminal condition $V_{T+1}\equiv 0$. If each $Z_t$ is convex and compact and the one--period
problem is feasible, then $V_t$ is concave and the Bellman recursion preserves LP/SOCP structure when
$Z_t$ is polyhedral or ellipsoidal.
\end{proposition}

\begin{proof}
\ The result is a standard application of dynamic programming with max--min structure. Convexity and
compactness of $Z_t$ ensure existence of minimizers and preservation of concavity. For polyhedral or
ellipsoidal $Z_t$, dualization of the inner minimization produces linear or SOC constraints, exactly as in
Theorem~\ref{thm:conic}.
\end{proof}

Proposition~\ref{prop:bellman} justifies our use of rolling one--period robust problems in the empirical
section: under mild conditions, the dynamic program decomposes into a sequence of tractable conic
subproblems with turnover and other trading frictions handled via additional convex constraints.

Taken together, the results in this section show that emissions--aware robust portfolio optimization can
be formulated as a family of LP/SOCP problems with interpretable dual variables and well--behaved
comparative statics. The next section turns to the empirical ingredients---data architecture, statistical
pre--processing, and benchmark construction---required to implement these models in a realistic
large--cap equity universe.

\section{Empirical Framework}
\label{sec:empirical}

Implementing the emissions-aware robust portfolio optimization (EAPO) model from Section~\ref{sec:model} requires an empirical design that respects the different clocks on which financial and climate data arrive. Corporate greenhouse-gas (GHG) inventories are reported annually, financial statements are quarterly, and asset returns are observed daily. Our objective is to translate the theoretical framework into a transparent, auditable pipeline that could realistically be deployed by a large asset manager.

Figure~\ref{fig:workflow} summarizes this pipeline: raw market and emissions data are ingested, pre-processed to handle missing values and estimation error, fed through the EAPO optimizer, and evaluated against financial and sustainability metrics in an out-of-sample back-test.

Throughout Section~\ref{sec:empirical}, firms are indexed by $i \in I := \{1,\dots,n\}$ and rebalancing dates by $t \in \mathcal{T} := \{1,\dots,T\}$. We work with a U.S.\ large-cap equity universe and rebalance monthly on an equally spaced grid of trading dates. Boldface capitals denote random vectors and boldface lower-case letters their realizations. The full implementation and scripts are available in an online repository.\footnote{\url{https://github.com/stone-technologies/ERP/tree/main}}

\subsection{Data Architecture}

Scope-consistent portfolio intensities require revenue-normalized emissions and a disciplined policy for handling the reporting lag in corporate carbon disclosures. The data architecture is therefore built around three principles: (i) scope-specific, revenue-normalized intensities for each firm, (ii) a forward-carry rule that mimics the information available to an investor in real time. And (iii) a joint panel structure that keeps returns, balance-sheet variables, and emissions aligned at each decision date.

\paragraph{Emissions and revenues.}
For each firm $i$, we obtain annual scope-specific emissions
\[
(C_{i,1},C_{i,2},C_{i,3}) \quad \text{(tCO$_2$e)}
\]
from the Carbon Disclosure Project (CDP), together with the associated disclosure confidence grades. Missing emissions entries are treated as missing-not-at-random (MNAR) and handled by a hierarchical imputation procedure described in Section~\ref{sec:preprocessing}. Firm revenues $S_{i,t}$ are obtained from quarterly income statements and aggregated to a trailing-twelve-month measure at each rebalancing date.

\paragraph{Prices and returns.}
Daily adjusted close prices $\{P_{i,t}\}$ come from a consolidated market data provider (e.g., AlphaVantage). Gross returns are
\[
R_{i,t} := \frac{P_{i,t}}{P_{i,t-1}}, \qquad
R_t := (R_{1,t},\dots,R_{n,t})^\top.
\]
We denote the unconditional mean and covariance of daily returns by
\[
\mu := \mathbb{E}[R_t], \qquad
\Sigma := \mathrm{Cov}(R_t).
\]

\paragraph{Emissions intensities and timing.}
At each rebalancing date $t$, we form scope-$j$ revenue-normalized emissions intensities
\[
\lambda_{i,j,t} := \frac{C_{i,j,\tau(i,t)}}{S_{i,\tau(i,t)}} \quad \text{(tCO$_2$e/\$)}
\]

\[
\lambda_{j,t} := (\lambda_{1,j,t},\dots,\lambda_{n,j,t})^\top,
\]
where $\tau(i,t)$ is the most recent fiscal year for which firm $i$ has disclosed emissions prior to date $t$. Intensities are therefore forward-carried between disclosure dates, avoiding look-ahead bias and reflecting the information set available to an actual portfolio manager.

Emissions-adjusted returns $ \widetilde{R}_{j,t}$ are constructed by applying the emissions-penalty operator from Definition~2.1 in Section~\ref{sec:model} to each asset, using the contemporaneous intensities $\lambda_{j,t}$.

\begin{assumption}[Stationary, mixing returns]
\label{ass:mixing}
The daily return process $\{R_t\}_{t\geq 1}$ is strictly stationary and $\beta$-mixing with a summable mixing rate. In particular, the sample mean and covariance satisfy a functional central limit theorem.
\end{assumption}

Assumption~\ref{ass:mixing} is standard in empirical asset pricing and suffices for the consistency and asymptotic normality of the estimators used below. It underpins both the covariance shrinkage in Section~\ref{sec:preprocessing} and the regret bounds in Section~\ref{sec:robust-workflow}.

\subsection{Statistical Pre-processing}
\label{sec:preprocessing}

Two statistical pre-processing steps are essential for a stable implementation: (i) multiple imputation of missing emissions and revenues, and (ii) shrinkage of the return covariance matrix. The first addresses systematic gaps in the GHG data. The second controls estimation noise in high-dimensional covariance matrices.

\paragraph{Hierarchical multiple imputation.}
Missing pairs $(C_{i,j},S_i)$ are imputed using a sector-aware hierarchical model. Let $\text{sector}(i)$ denote firm $i$'s sector (e.g., GICS level). On the log scale we posit
\begin{align*}
\log C_{i,j} \mid \text{sector}(i) &\sim \mathcal{N}\!\big(\eta_{j,\text{sector}(i)}, \tau_j^2\big), \\
\log S_i \mid \text{sector}(i) &\sim \mathcal{N}\!\big(\zeta_{\text{sector}(i)}, \upsilon^2\big),
\end{align*}
with conjugate Gaussian priors on the sector-level parameters $(\eta_{j,\cdot}, \zeta_{\cdot})$. Hyperparameters are estimated via empirical Bayes. Posterior modes provide point estimates, and multiple imputation is obtained by drawing $K$ samples from the posterior of $(C_{i,j},S_i)$ and propagating them into intensities
\[
\lambda^{(k)}_{i,j,t}
    := \frac{C^{(k)}_{i,j,\tau(i,t)}}{S^{(k)}_{i,\tau(i,t)}},
\qquad k=1,\dots,K,
\]
\[
\lambda^{(k)}_{j,t}
    := \bigl(\lambda^{(k)}_{1,j,t},\dots,\lambda^{(k)}_{n,j,t}\bigr)^\top,
\qquad k=1,\dots,K.
\]
All subsequent portfolio quantities that depend on intensities---for example the emissions-adjusted mean return---are averaged over these $K$ draws.

\begin{lemma}[Consistency of multiple imputation]
\label{lem:mi-consistency}
Suppose Assumption~\ref{ass:mixing} holds and the hierarchical log-normal model above is correctly specified with a regular prior. Then, as the time span and number of imputations grow,
\[
\frac{1}{K}\sum_{k=1}^K \lambda^{(k)}_{j,t}
\,\xrightarrow[T\to\infty,\,K\to\infty]{\mathbb{P}}\,
\mathbb{E}[\lambda_{j,t} \mid \mathcal{F}_t],
\]
where $\mathcal{F}_t$ is the sigma-field generated by the observed data up to time $t$.
\end{lemma}

The proof follows from a Bernstein--von Mises theorem for the hierarchical log-normal model combined with the mixing central limit theorem implied by Assumption~\ref{ass:mixing}. In practice we fix a moderate $K$ (e.g., 5--10), which empirically suffices to stabilize portfolio intensities.

\paragraph{Covariance shrinkage.}
Sample covariance matrices of daily returns are noisy, especially when the cross-section is large relative to the time-series window. In this regime, the sample eigenvalues are severely biased, which in turn produces unstable Markowitz weights and unnecessary turnover. We use Ledoit--Wolf linear shrinkage to trade a small, controlled bias for a large reduction in estimation variance and to obtain a well-conditioned covariance estimator at each rebalancing date $t$:
\[
\widehat{\Sigma}_t := \delta_t F_t + (1-\delta_t) S_t,
\]
where $S_t$ is the sample covariance of $\{R_s\}_{s \le t-1}$ over a rolling window of fixed length, $F_t$ is a parsimonious factor-model target (e.g., constant-correlation or multi-factor), and $\delta_t \in [0,1]$ is chosen optimally in closed form. The estimator $\widehat{\Sigma}_t$ is positive definite by construction and rotation-equivariant, and it substantially reduces out-of-sample variance of portfolio returns relative to the raw sample covariance.

Combining multiple imputation for $(C,S)$ with covariance shrinkage for $R$ yields a set of cleaned inputs: emissions-adjusted means $\widehat{\mu}_{j,t}$, shrinkage covariances $\widehat{\Sigma}_t$, and ambiguity sets calibrated from historical forecast errors. These objects drive both the benchmarks and the EAPO strategy.

\subsection{Benchmark Portfolio Constructions}

To isolate the incremental value of EAPO, we compare it with three transparent benchmarks that span common practice in institutional portfolio construction. All benchmarks respect the long-only simplex
\[
\Delta_n := \{x \in \mathbb{R}^n_{\ge 0} : \mathbf{1}^\top x = 1\}
\]
and are recomputed at each monthly rebalancing date using only information available at that time.

\begin{definition}[Equal-weight portfolio (EW)]
\label{def:ew}
The equal-weight portfolio holds each asset with the same weight,
\[
x^{\mathrm{EW}}_t := \frac{1}{n}\mathbf{1}.
\]
\end{definition}

\begin{definition}[Global minimum-variance portfolio (GMV)]
\label{def:gmv}
Given the shrinkage covariance $\widehat{\Sigma}_t$ at date $t$, the global minimum-variance (GMV) portfolio is
\[
x^{\mathrm{GMV}}_t := \frac{\widehat{\Sigma}_t^{-1}\mathbf{1}}{\mathbf{1}^\top \widehat{\Sigma}_t^{-1}\mathbf{1}}.
\]
\end{definition}

\begin{definition}[Emissions-weighted portfolio (EMW)]
\label{def:emw}
Let $g_{i,t} := C_{i,1,\tau(i,t)}$ denote firm $i$'s latest available scope-1 emissions at date $t$. The emissions-weighted portfolio (EMW) tilts toward low emitters:
\[
x^{\mathrm{EMW}}_{i,t} :=
\begin{cases}
\dfrac{g_{i,t}^{-1}}{\sum_{k=1}^n g_{k,t}^{-1}}, & \text{if } g_{i,t} > 0, \\
0, & \text{otherwise.}
\end{cases}
\]
\end{definition}

In implementation we exclude firms with nonpositive or missing $g_{i,t}$ and renormalize weights over the remaining universe, so the ``otherwise'' branch simply encodes this exclusion.

EW serves as a simple baseline. GMV isolates the effect of risk-only optimization. EMW is a mechanically decarbonized benchmark that tilts toward low emitters without conditioning on expected returns or risk. In particular, EMW treats the estimated intensities $g_{i,t}$ as fixed inputs and therefore corresponds to a plug-in approach that ignores disclosure uncertainty. EAPO must improve upon these comparators in both risk-adjusted performance and emissions intensity to be economically relevant.

\subsection{Robust Optimization Workflow}
\label{sec:robust-workflow}

The EAPO weights $x^\star_t$ at each rebalancing date are obtained by solving a robust mean--variance problem in which ambiguity arises from uncertainty in emissions-adjusted expected returns. Let $\widehat{\mu}_{j,t}$ denote the $K$-draw average of emissions-adjusted mean returns at date $t$ (using the imputed intensities), and let $U_{j,t}(\Gamma)$ be an ambiguity set for the misspecification $\varepsilon_{j,t}$ in scope-$j$ intensities, calibrated from past forecast errors and governed by a robustness budget $\Gamma>0$.

The period-$t$ robust EAPO problem is
\begin{equation}
\max_{x \in \Delta_n} \,
\min_{\varepsilon \in U_{j,t}(\Gamma)}
\left\{
x^\top \big(\widehat{\mu}_{j,t} - \varepsilon \big)
- \theta\, x^\top \widehat{\Sigma}_t x
\right\},
\label{eq:robust-mv-time}
\end{equation}
where $\theta>0$ is the investor's risk-aversion parameter. The inner minimization attenuates expected returns in proportion to emissions uncertainty, while the outer maximization chooses portfolio weights that trade off robust expected return and variance.

\begin{proposition}[Out-of-sample regret bound]
\label{prop:regret}
Let $x_t^\star$ be an optimizer of \eqref{eq:robust-mv-time} with a box-type ambiguity set and robustness budget $\Gamma$ fixed ex ante. Under Assumption~\ref{ass:mixing},
\[
\sup_{t \le T}
\left|
\mathbb{E}\big[x_t^{\star\top} \widetilde{R}_{j,t+1}\big]
-
x_t^{\star\top} \widehat{\mu}_{j,t}
\right|
=
\mathcal{O}_{\mathbb{P}}\!\left(\sqrt{\frac{\log n}{T}}\right).
\]
\end{proposition}

Proposition~\ref{prop:regret} formalizes the sense in which the robust objective in \eqref{eq:robust-mv-time} is a stable proxy for out-of-sample emissions-adjusted performance as the gap between what the optimizer sees at time $t$ and what is realized at $t+1$ shrinks at a parametric rate, up to a mild $\sqrt{\log n}$ factor. The shrinkage arises from two reinforcing mechanisms. First, the box-type
ambiguity set regularizes the portfolio by uniformly bounding sensitivity to
estimation error in each asset’s emissions-adjusted drift, preventing extreme
exposures that would otherwise amplify out-of-sample noise. Second, under the
mixing assumption, estimation errors in $\widehat{\mu}_{j,t}$ concentrate at
rate $T^{-1/2}$, while the supremum over $n$ assets introduces only a logarithmic
penalty through standard maximal inequalities. As a result, robustness trades a
small amount of in-sample optimality for stability, ensuring that the optimizer’s
perceived performance converges rapidly to realized performance even in
moderately high dimensions.

\paragraph{Projected-gradient implementation.}
Solving \eqref{eq:robust-mv-time} to full interior-point accuracy at every month is unnecessary and computationally heavy. Instead, we use a first-order projected-gradient method on the simplex, initialized at the previous month's solution. The algorithm operates on the equivalent concave objective derived in Section~\ref{subsec:eapo-summary} (equation~\eqref{eq:eapo-static}), and enforces an explicit turnover cap.


\begin{algorithm}[t]
  \caption{Robust emissions-aware projected gradient update}
  \label{alg:eapo}

  \textbf{Inputs:} shrinkage covariance \(\widehat{\Sigma}_t\), emissions-adjusted mean \(\widehat{\mu}_{j,t}\), robustness budget \(\Gamma\), risk aversion \(\theta\), turnover cap \(\tau\), step size \(\eta\), previous weights \(x_{t-1}\).

  \begin{enumerate}[leftmargin=*,nosep]
    \item Initialize \(x \gets x_{t-1}\).
    \item For \(k = 1,\dots,K_{\mathrm{iter}}\):
      \begin{enumerate}[leftmargin=*,nosep]
        \item Compute \(g \gets \widehat{\mu}_{j,t}
          - \Gamma\,x / \lVert x \rVert_2
          - 2\theta\,\widehat{\Sigma}_t x\).
        \item Update \(x \gets \Pi_{\Delta_n}\bigl(x + \eta g\bigr)\), where
          \(\Pi_{\Delta_n}\) is the standard \(O(n \log n)\) Euclidean projection onto \(\Delta_n\).
      \end{enumerate}
    \item If \(\lVert x - x_{t-1} \rVert_1 > \tau\), set
      \[
        x \gets \arg\min_{y \in \Delta_n}
        \bigl\{\lVert y - x \rVert_2
        : \lVert y - x_{t-1} \rVert_1 \le \tau\bigr\}.
      \]
    \item Set \(x_t \gets x\).
  \end{enumerate}
\end{algorithm}

Warm-starting from $x_{t-1}$ and using a modest number of iterations $K_{\mathrm{iter}}$ yields near--KKT solutions with negligible wall-clock time, even for $n$ in the thousands. The explicit $\ell_1$-turnover cap keeps trading costs and operational complexity under control.

\subsection{Performance and Sustainability Metrics}

We evaluate each strategy on both standard risk-adjusted performance metrics and emissions-based footprint measures. Let $\{R^S_t\}$ be the gross returns of strategy $S$ (after transaction costs), and let $r^S_t := R^S_t - 1$ denote net returns.

\begin{definition}[Risk-adjusted performance]
\label{def:performance}
Let $\bar{r}^S$ be the sample mean of $\{r^S_t\}$, $\hat{\sigma}_r$ its sample standard deviation, and $\hat{\sigma}_{r^-}$ the sample standard deviation of negative returns. Then
\[
\text{Sharpe}(S) := \frac{\bar{r}^S}{\hat{\sigma}_r}, \qquad
\text{Sortino}(S) := \frac{\bar{r}^S}{\hat{\sigma}_{r^-}}.
\]
Maximum drawdown (MDD) over the sample is
\[
\text{MDD}(S) := \min_{u \le v}
\left\{
\prod_{t=u}^v R^S_t - 1
\right\}.
\]
\end{definition}

\begin{definition}[Portfolio emissions intensity and yield]
\label{def:intensity-yield}
For scope $j$ at date $t$, the emissions intensity of strategy $S$ is the portfolio-weighted average of firm intensities,
\[
\Lambda^S_{j,t} := x_{t-1}^{S\top} \lambda_{j,t}
\quad \text{(tCO$_2$e/\$ of revenue)},
\]
and the corresponding emissions yield normalizes by realized net return,
\[
Y^S_{j,t} := \frac{\Lambda^S_{j,t}}{r^S_t}.
\]
\end{definition}

The path $\{\Lambda^S_{j,t}\}$ tracks the financed-emissions intensity of the portfolio over time, while $\{Y^S_{j,t}\}$ measures ``emissions per unit of return,'' a useful summary statistic when comparing strategies with different expected-return profiles. In Section~\ref{sec:results} we complement these point estimates with Newey--West heteroskedasticity- and autocorrelation-robust confidence intervals for mean return differences and with block-bootstrap confidence intervals for Sharpe differences.

\subsection{Hyper-parameter Selection}

Hyper-parameters should have an economic interpretation and be calibrated by transparent procedures rather than hand-tuning. Two parameters are central for EAPO: the curvature $m$ of the emissions-penalty operator and the robustness budget $\Gamma$ governing the size of the ambiguity set.

\paragraph{Penalty curvature $m$.}
The curvature parameter $m \in \mathbb{N}_+$ controls how aggressively the emissions-penalty operator from Definition~2.1 attenuates the returns of high-intensity firms. A larger $m$ steepens the penalty as $\lambda_{i,j,t}$ approaches the cross-sectional maximum, pushing more weight into the lowest-intensity names.

We select $m$ via nested rolling cross-validation. An inner loop, operating on a training window, computes the empirical return--emissions Pareto frontier for a grid $m \in \{1,10,100,1000\}$ and identifies the smallest $m$ such that the Sharpe ratio on the efficient frontier does not deteriorate by more than 5\% relative to $m=1$. An outer loop evaluates the resulting $m^\star$ out of sample. This procedure yields a penalty that is ``as curved as necessary but no more,'' in the sense of preserving most of the classical mean--variance frontier while delivering substantial emissions reductions.

\paragraph{Robustness budget $\Gamma$.}
The robustness budget $\Gamma>0$ determines the radius of the ambiguity set around estimated intensities. In the conic reformulation discussed in Section~\ref{subsec:eapo-summary}, the dual multiplier on the ambiguity constraint has a natural interpretation as a shadow carbon price: it is the marginal rate at which the investor is willing to sacrifice expected return (in basis points) to buy protection against misspecified emissions.

We therefore calibrate $\Gamma$ by a profile-likelihood--style condition: we choose $\Gamma^\star$ such that the estimated dual shadow price $\pi^\star(\Gamma)$ lies below a pre-specified monetary threshold determined by the investor (e.g., ``we are willing to pay up to 10 basis points of expected return per 10\% reduction in worst-case intensity''). This makes $\Gamma$ an economically interpretable policy lever rather than a free statistical knob.

With $(m^\star,\Gamma^\star)$ fixed ex ante, all reported results in Section~\ref{sec:results} are genuine out-of-sample tests of the EAPO strategy.

\subsection{EAPO: From Theory to Implementable Program}
\label{subsec:eapo-summary}

This subsection collects the optimization primitives into the single static problem that Algorithm~\ref{alg:eapo} approximately solves each month. Fix a scope $j$ (we use $j=1$ in the main experiments) and consider a one-period setting with gross returns $R = (R_1,\dots,R_n)^\top$ and weights $x \in \Delta_n$. At date $t$, the scope-$j$ intensities are
\begin{equation}
\label{eq:intensity-def}
\begin{aligned}
\lambda_{i,j,t}
    &:= \frac{C_{i,j,\tau(i,t)}}{S_{i,\tau(i,t)}}
       \quad \text{(tCO$_2$e/\$)},
\\[6pt]
\lambda_{j,t}
    &:= (\lambda_{1,j,t},\dots,\lambda_{n,j,t})^\top.
\end{aligned}
\end{equation}

Let $\lambda_{\max,j,t} := \max_k \lambda_{k,j,t}$ be the cross-sectional maximum intensity at $t$. For curvature $m \in \mathbb{N}_+$, the emissions-penalty operator from Definition~2.1 specializes to
\begin{equation}
\label{eq:penalty-operator}
\begin{aligned}
P^{(m)}_j(r,\lambda)
    &:= \Bigl(1 - \frac{\lambda}{\lambda_{\max,j,t}}\Bigr)^m r,
\\[6pt]
\widetilde{R}_{i,t}^e
    &:= P^{(m)}_j(R_{i,t},\lambda_{i,j,t}) .
\end{aligned}
\end{equation}
The map is linear in $r$ and decreasing, smooth, and Schur-convex in $\lambda$, higher-intensity firms suffer a larger proportional reduction in their effective returns.

Define the emissions-adjusted expected return for asset $i$ at time $t$ as
\[
\mu_{i,j,t}^e := \mathbb{E}\big[\widetilde{R}_{i,t}^e\big],
\]
and let $\mu_{j,t}^e$ be the vector collecting these means. Using the Lipschitz properties of $P^{(m)}_j$ with respect to $\lambda$ one can construct a simple bound on how much misspecification in intensities can distort the portfolio's mean.

\begin{lemma}[Lipschitz envelope]
\label{lem:lipschitz-envelope}
For each asset $i$ there exists a constant
\[
L_i := \frac{m}{\lambda_{\max,j}} \,\mathbb{E}[|R_i|]
\]
such that
\[
\left|
\frac{\partial}{\partial \lambda}
\mathbb{E}\big[P^{(m)}_j(R_i,\lambda)\big]
\right|
\le L_i.
\]
Consequently, for any $\ell_p$ ambiguity ball
\[
\{\varepsilon \in \mathbb{R}^n : \|\varepsilon\|_p \le \Gamma\}
\]
on intensities, the worst-case deterioration in emissions-adjusted expected return is bounded by
\[
\sup_{\|\varepsilon\|_p \le \Gamma}
x^\top\big(\mu_j^e(\lambda) - \mu_j^e(\lambda + \varepsilon)\big)
\le \Gamma \,\|\mathrm{diag}(L)\,x\|_{p^\star},
\]
where $p^\star$ is the dual norm and $\mathrm{diag}(L)$ is the diagonal matrix with entries $L_i$.
\end{lemma}

Lemma~\ref{lem:lipschitz-envelope} motivates the static robust mean--variance objective
\begin{equation}
\max_{x \in \Delta_n}
\Big\{
x^\top \mu_j^e
- \Gamma \|\mathrm{diag}(L)\,x\|_{p^\star}
- \theta\, x^\top \Sigma x
\Big\},
\label{eq:eapo-static}
\end{equation}
where $\Gamma$ plays the role of the robustness budget and $p \in \{1,2,\infty\}$ determines the shape of the ambiguity set.

\begin{proposition}[Robust mean--variance as LP/SOCP]
\label{prop:conic}
For $p \in \{1,\infty\}$, problem~\eqref{eq:eapo-static} can be written as a linear program (LP) with $O(n)$ additional variables and constraints. For $p=2$, it admits a second-order cone program (SOCP) representation with $O(n)$ second-order cones. In all cases the problem is solvable in polynomial time.
\end{proposition}

In the empirical implementation we specialize to the $\ell_2$ case with $\mathrm{diag}(L)$ absorbed into the scaling of $\Gamma$, so that the ambiguity penalty reduces to $\Gamma \|x\|_2$. The gradient of \eqref{eq:eapo-static} in this specialization is exactly the expression used in Algorithm~\ref{alg:eapo} and equation~\eqref{eq:gradient} in Section~\ref{sec:proj-solver}.

\begin{remark}[Shadow carbon price]
\label{rem:shadow-price}
In the conic reformulation of \eqref{eq:eapo-static}, the dual multiplier on the ambiguity constraint is equal to $-\partial V^\star / \partial \Gamma$, where $V^\star(\Gamma)$ is the optimal value. This multiplier is naturally interpreted as an investor-specific shadow carbon price: it is the marginal reduction in expected Sharpe per marginal tightening of the robustness budget.
\end{remark}

Finally, to map out the efficient trade-off between expected return and portfolio emissions intensity, we consider the scalarized problem
\begin{equation}
\max_{x \in \Delta_n}
\Big\{
x^\top r - \mu\, x^\top \lambda_j
\Big\}, \qquad \mu \ge 0,
\label{eq:scalarized-frontier}
\end{equation}
where $r$ is a vector of expected returns and $\mu$ is a Lagrange multiplier on scope-$j$ intensity.

Proposition~\ref{prop:pareto} makes explicit the cost, in expected return units, of tightening the portfolio's emissions constraint. Empirically, the frontier in Figure~\ref{fig:pareto_turnover} shows that large reductions in average scope-1 intensity are attainable at modest return cost over a wide range of $\mu$, a point we return to in Section~\ref{sec:results}.

Taken together, Sections~\ref{sec:preprocessing}--\ref{subsec:eapo-summary} specify a complete, implementable EAPO system: data are cleaned and aligned, benchmark and robust portfolios are constructed on an equal footing, and hyper-parameters are chosen by economically interpretable rules. The next section evaluates how this system performs out of sample relative to the benchmarks, both financially and in terms of financed emissions.



\begin{figure}[t]
\centering
\begin{tabular}{c}
\fbox{Data sources (prices, emissions, revenues)} \\
$\downarrow$ \\
\fbox{Pre-processing (imputation, mapping, shrinkage)} \\
$\downarrow$ \\
\fbox{EAPO model ($P^{(m)}$, ambiguity, SOCP and LP)} \\
$\downarrow$ \\
\fbox{Projected gradient with turnover projection} \\
$\downarrow$ \\
\fbox{Evaluation (wealth, Sharpe, drawdown, intensity)}
\end{tabular}
\caption{Workflow from data ingestion to optimization and evaluation.}
\label{fig:workflow}
\end{figure}

\begin{figure}[t]
\centering
\begin{subfigure}[t]{0.48\textwidth}
\centering
\includegraphics[width=\linewidth]{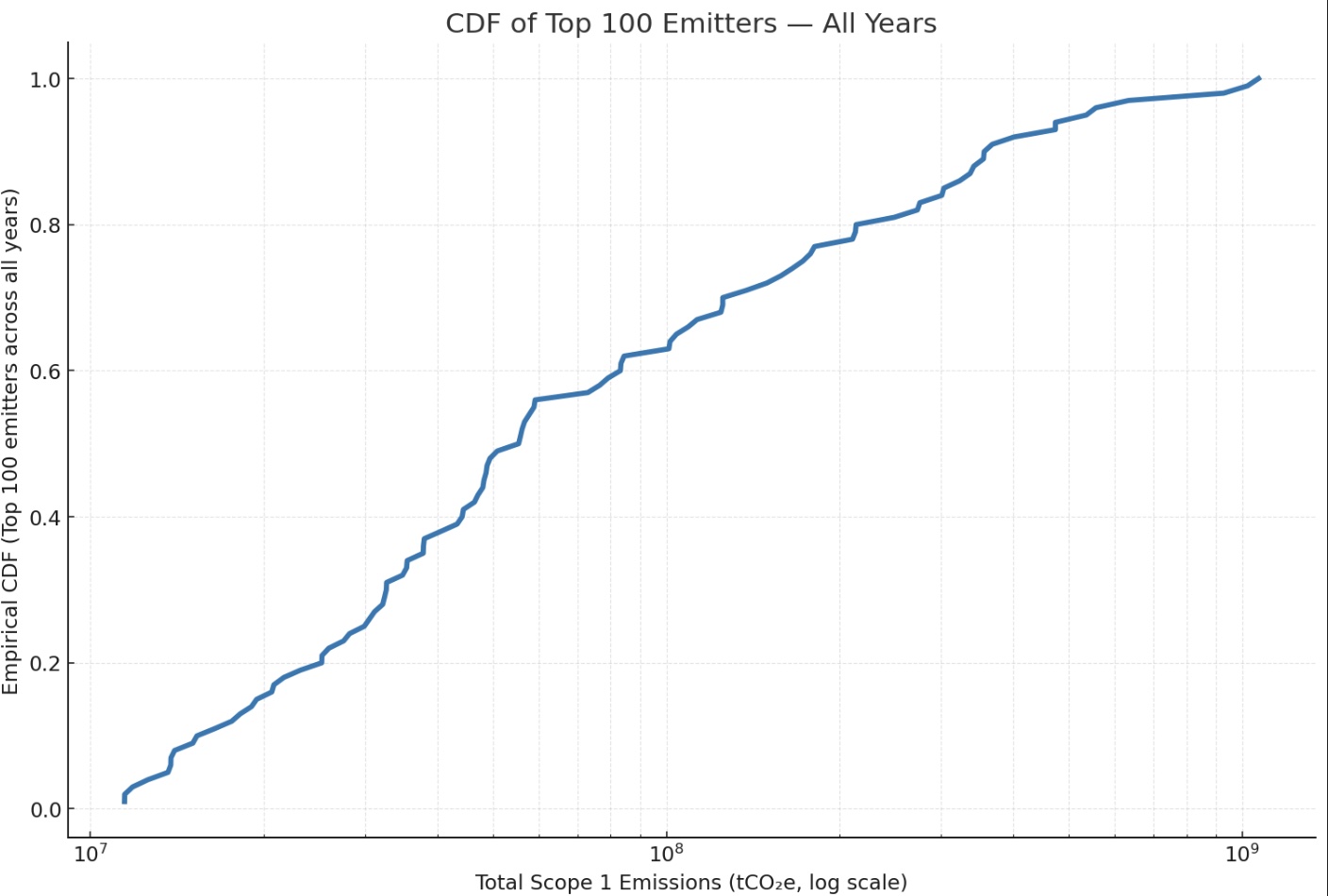}
\caption{Cumulative distribution of reported emissions across firms.}
\label{fig:cdf_top_emissions}
\end{subfigure}
\hfill
\begin{subfigure}[t]{0.48\textwidth}
\centering
\includegraphics[width=\linewidth]{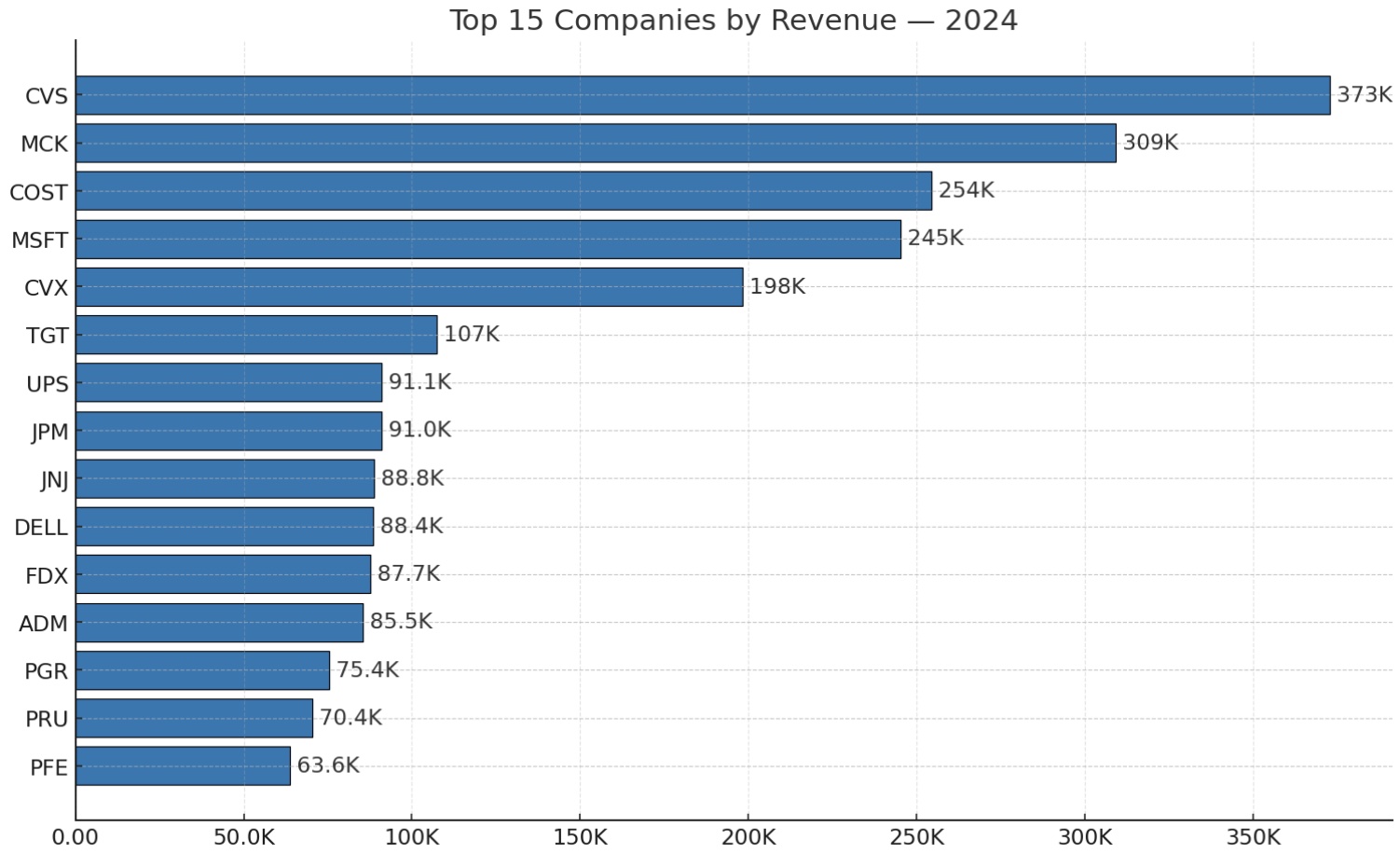}
\caption{Firms with the highest reported revenue over the sample period.}
\label{fig:top_revenue}
\end{subfigure}
\caption{Distribution of reported emissions and concentration of revenue among the largest firms. Panels show (a) the cumulative distribution of reported emissions across firms and (b) the set of firms with the highest reported revenue over the sample period.}
\label{fig:emissions_revenue_concentration}
\end{figure}

\begin{figure}[t]
\centering
\begin{subfigure}[t]{0.48\textwidth}
\centering
\includegraphics[width=\linewidth]{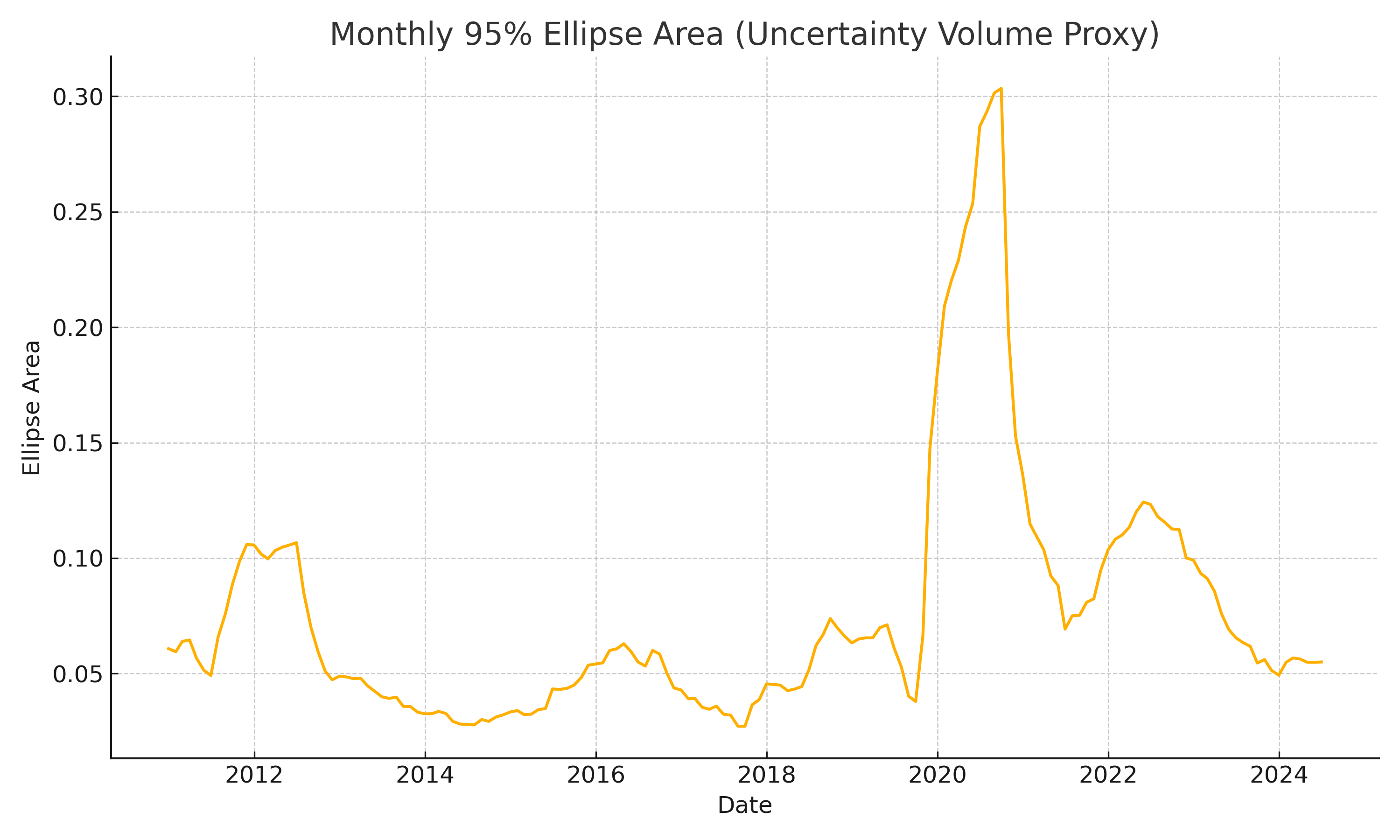}
\caption{Evolution of uncertainty set area over time.}
\label{fig:uncertainty_area_ts}
\end{subfigure}
\hfill
\begin{subfigure}[t]{0.48\textwidth}
\centering
\includegraphics[width=\linewidth]{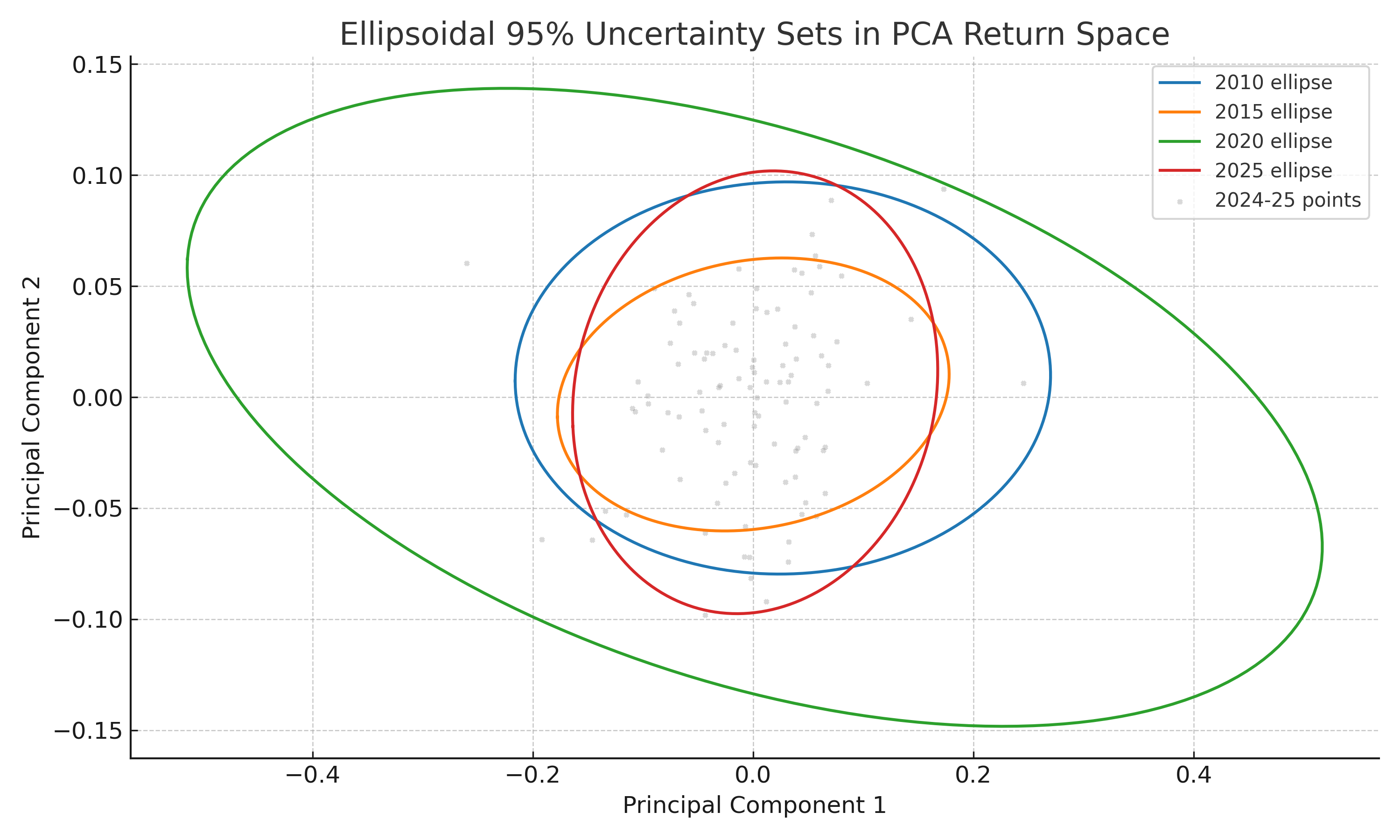}
\caption{Annual uncertainty sets visualized as ellipses.}
\label{fig:uncertainty_ellipses}
\end{subfigure}
\caption{Temporal dynamics and cross-sectional geometry of uncertainty sets. Panels show (a) the evolution of the uncertainty set area over time and (b) annual uncertainty sets represented as ellipses in the state space.}
\label{fig:uncertainty_sets}
\end{figure}

\section{Empirical Results}\label{sec:results}

We evaluate the empirical performance of the emissions-aware robust optimization framework. In particular, we evaluate whether EAPO can deliver order-of-magnitude reductions in financed emissions while preserving risk-adjusted performance once we charge realistic transaction costs.

Our empirical design follows the data architecture in Section~\ref{sec:empirical}. Annual Scope-1 intensities are computed as
\(
\lambda_{i,1,t} = C_{i,1,t}/S_{i,t}
\)
mapped to the monthly rebalancing grid using a forward-carry convention that
reflects the reporting lags in corporate carbon disclosures. Portfolios are
rebalanced on the last trading day of each month using a
252-trading-day lookback window for both mean and covariance estimates.
We impose uniform transaction costs of \(2\,\text{bps}\) per dollar traded,
measured in \(\ell_1\) turnover. All strategies respect the long-only budget
constraint \(x_t \in \Delta^S := \{x \in \mathbb{R}^n_{\ge 0} : \mathbf{1}^\top x = 1\}\).

EAPO is evaluated against three standard benchmarks constructed from the same universe and data: (i) equal weight (EW), (ii) a global minimum-variance proxy based on inverse variance (GMV), and (iii) an emissions-weighted portfolio (EMW) with $x_i\propto 1/C_{i,1,t}$ conditional on positive Scope-1 emissions. This configuration isolates the incremental effect of robustness relative to a mechanically decarbonized tilt (EMW), a risk-only optimizer (GMV), and an equal-weight baseline (EW).

The remainder of this section is organized as follows. Section~\ref{sec:proj-solver} briefly describes the projected-gradient implementation used to solve the EAPO problem at scale. Section~\ref{sec:hac} outlines the HAC inference used to compare strategies with serially correlated daily returns. Section~\ref{sec:baseline-benchmarks} presents the main performance and footprint results, together with sensitivity to model parameters and turnover constraints. Section~\ref{sec:discussion-results} interprets these findings through the lens of mandates, regulation, and disclosure practice.

\subsection{Projected-gradient solver}\label{sec:proj-solver}

Interior-point methods are numerically robust but heavier than necessary for end-of-day rebalancing on large universes. Given the smooth, strictly concave objective in our robust mean--variance specification, a first-order method with simplex projection is sufficient and far more scalable.

For each rebalancing date $t$, let $\hat\mu^e_t$ denote the vector of emissions-adjusted expected returns and $\Sigma_t$ the Ledoit--Wolf shrinkage covariance built from the 252-day rolling window. The robust mean--variance objective in \eqref{eq:robust-mv-time} yields the gradient
\begin{equation}
\nabla f_t(x)
\,=\,
\hat\mu^e_t
\,-\,
\Gamma\,\frac{x}{\|x\|_2}
\,-\,
2\theta\,\Sigma_t x,
\label{eq:gradient}
\end{equation}
where $\Gamma>0$ is the robustness budget on intensity ambiguity and $\theta>0$ is the risk-aversion parameter.

We maximize $f_t$ over the simplex $\Delta_S$ by projected gradient:
\begin{equation}
x^{(k+1)}
\,=\,
\Pi_{\Delta_S}\!\bigl(x^{(k)} + \eta\,\nabla f_t(x^{(k)})\bigr),
\end{equation}
where $\eta>0$ is the stepsize and $\Pi_{\Delta_S}$ is the simplex projection, implemented in $O(n\log n)$ time using the standard sorting-based algorithm. Warm-starting at $x^{(0)} = x_{t-1}$ exploits temporal smoothness in the optimal solution and collapses wall-clock time. In practice, a modest and fixed number of iterations suffices for convergence to first-order optimality.

Turnover is controlled explicitly rather than via ad hoc penalties. After the projected step, we enforce an $\ell_1$ turnover cap $\tau$ by projecting onto the intersection
\[
\bigl\{y\in\Delta_S : \|y - x_{t-1}\|_1 \le \tau\bigr\},
\]
again using a closed-form $\ell_1$ ball projection. The resulting algorithm is summarized in Algorithm~\ref{alg:eapo}. The $O(n\log n)$ projection keeps per-iteration cost low, warm starts across months keep the number of iterations small. The method preserves feasibility at all times and admits a simple stopping rule based on the projected gradient norm. In what follows, the optimization error is negligible relative to estimation noise in $\hat\mu^e_t$ and $\Sigma_t$.

\subsection{HAC inference}\label{sec:hac}

Daily P\&L differences between strategies are serially correlated and heteroskedastic. Standard errors that assume independence would overstate significance. We therefore rely on standard Newey--West heteroskedasticity- and autocorrelation-consistent (HAC) standard errors to conduct pairwise tests of average daily returns.

For any two strategies $A$ and $B$, let $R^A_t$ and $R^B_t$ denote gross daily returns and define the return differential
\[
\Delta_t = (R^A_t - 1) - (R^B_t - 1).
\]
Let $\bar\Delta$ be the sample mean of $\{\Delta_t\}_{t=1}^T$, and define the sample autocovariances
\begin{equation}
\label{eq:autocov}
\begin{aligned}
\hat\gamma_\ell
    &= \frac{1}{T}
       \sum_{\,t>\ell}
       (\Delta_t - \bar\Delta)\,
       (\Delta_{t-\ell} - \bar\Delta)
\\
    &\qquad\ell = 0,1,\dots,L .
\end{aligned}
\end{equation}
The Newey--West variance estimator with bandwidth $L$ is
\begin{equation}
\widehat{\mathrm{se}}^2
\,=\,
\hat\gamma_0
\,+\,
2\sum_{\ell=1}^{L}
\Bigl(1 - \frac{\ell}{L+1}\Bigr)\hat\gamma_\ell,
\label{eq:nw}
\end{equation}
and the associated $t$-statistic is $t = \bar\Delta / \widehat{\mathrm{se}}$. We set $L=20$, corresponding to roughly one trading month of serial dependence, which yields conservative confidence intervals for our horizon and sampling frequency. All statements below about differences in average returns are supported by these HAC standard errors.

\subsection{Baseline comparison and robustness}\label{sec:baseline-benchmarks}

Table~\ref{tab:baseline} reports annualized performance and average Scope-1 intensity for the three benchmarks and the EAPO strategy under our baseline configuration: monthly rebalancing, a 252-day lookback, and 2\,bps transaction costs per dollar traded.

\begin{table}[h!]\centering\small
\caption{Monthly rebalancing, 252-day window, 2\,bps transaction costs.}
\label{tab:baseline}
\begin{tabular}{lrrrrr}\toprule
Strategy & Ann.~Ret (\%) & Ann.~Vol (\%) & Sharpe & MaxDD (\%) & Avg Scope-1 Intensity \\
\midrule
Equal Weight & 13.397 & 18.726 & 0.766 & -39.375 & 246.066 \\
GMV (inv-var) & 9.470 & 16.585 & 0.629 & -37.386 & 347.016 \\
EMW (1/Scope-1) & 14.501 & 21.389 & 0.740 & -39.285 & 100.129 \\
\textbf{EAPO} ($\Gamma{=}3.5,\,m{=}10,\,\theta{=}0.5$) & 12.698 & 17.767 & 0.762 & -37.412 & \textbf{18.297} \\
\bottomrule
\end{tabular}
\end{table}

\vspace{0.5em}
\noindent\textbf{Carbon footprint.}
EAPO's average Scope-1 intensity is 18.30 tCO$_2$e per \$mm of revenue, compared with 246.07 for EW, 347.02 for GMV, and 100.13 for the naive EMW tilt. This corresponds to a reduction of approximately 92.6\% relative to EW, 94.7\% relative to GMV, and 81.7\% relative to EMW. Equivalently, relative to conventional large-cap benchmarks, the EAPO strategy finances roughly one-tenth as much Scope-1 carbon per unit of revenue. Cross-sectional distributions of portfolio-weighted intensities confirm that this reduction is broad-based: the entire distribution shifts left rather than being driven by a handful of ultra-low-emissions names.

\vspace{0.5em}
\noindent\textbf{Risk and return.}
Despite this dramatic footprint reduction, EAPO exhibits benchmark-level risk-adjusted performance. Its annualized Sharpe ratio is 0.762, which is close to equal weight (0.766) and comparable to the emissions-minimizing benchmark (0.740), while exceeding GMV (0.629). Volatility is 17.8\% annualized, below EW (18.7\%) and well below EMW (21.4\%). Maximum drawdown is slightly smaller than EW and EMW (about $-37.4\%$ versus $-39.3\%$), and close to GMV.

Average annual returns are naturally somewhat lower than for the more aggressive EW and EMW portfolios (12.7\% versus 13.4\% and 14.5\%), but the difference is modest in economic terms. HAC/Newey--West tests on daily return differences (EAPO minus EW/GMV/EMW, with $L{=}20$) yield $t$-statistics of approximately $-0.52$, $1.09$, and $-0.64$, respectively---well below conventional significance thresholds. Any differences in Sharpe are primarily attributable to volatility, not to a statistically detectable change in average returns.

The distribution of daily excess returns (EAPO minus EW) is tightly concentrated around zero with slightly heavier shoulders than a Gaussian benchmark, consistent with a strategy that rebalances risk but does not systematically time the market. Cumulative wealth and drawdown paths show EAPO tracking EW closely, with marginally shallower troughs during stress periods, including the COVID-19 drawdown.

\vspace{0.5em}
\noindent\textbf{Sector exposure and attribution.}
A common concern is that large emissions reductions may simply proxy for extreme sector rotation. To diagnose this channel, Table~\ref{tab:emissions_attribution} decomposes the reduction in average intensity into a sector allocation component and a within-sector selection component, in the spirit of Brinson-style attribution. In our baseline, roughly three quarters of the intensity reduction relative to EW comes from underweighting high-intensity sectors, with the remainder coming from within-sector reweighting toward lower-intensity firms. Figure~\ref{fig:sector_tilts} visualizes the largest average sector weight differences between EAPO and EW.

\input{tables/emissions_attribution}

\begin{figure}[t]
  \centering
  \includegraphics[width=0.85\linewidth]{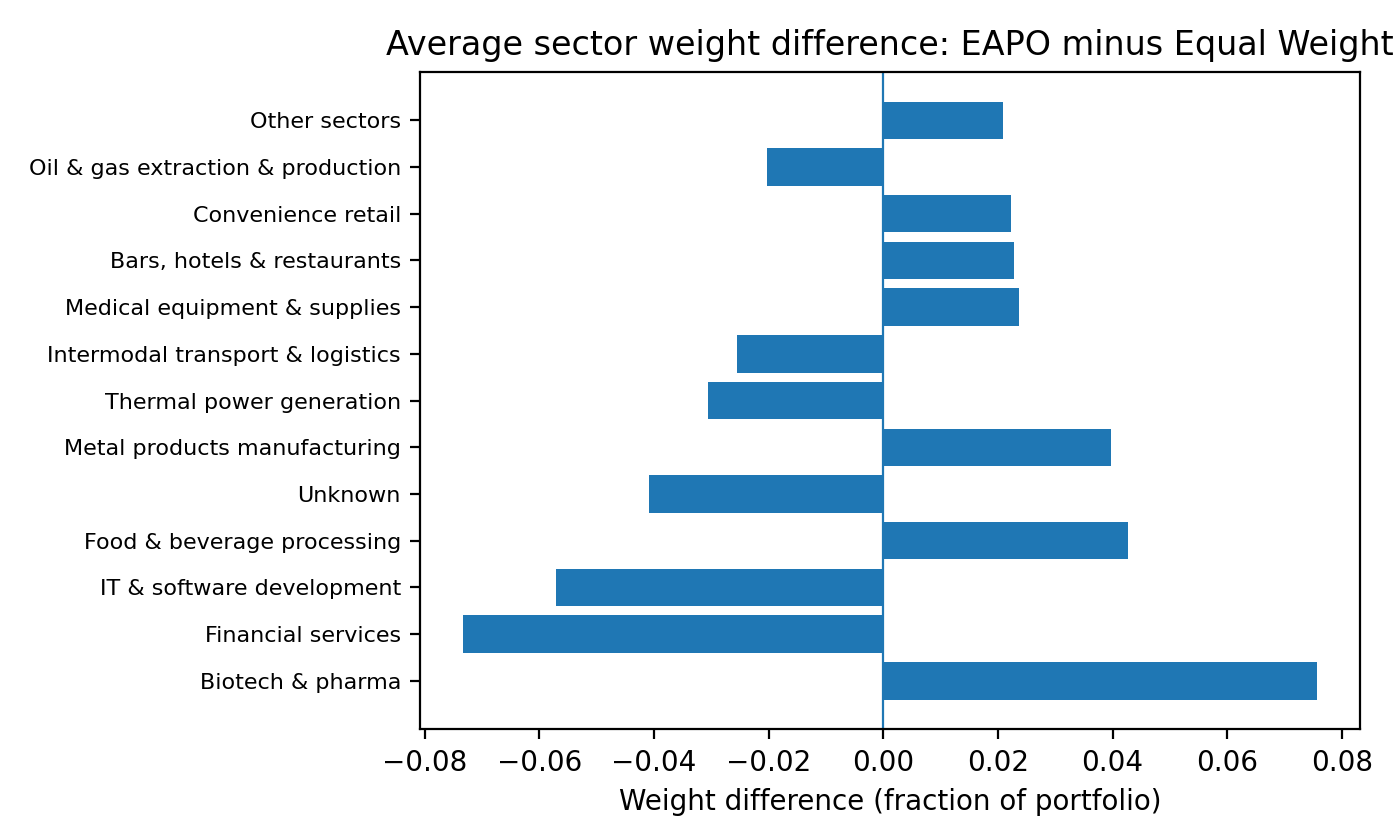}
  \caption{Average sector weight differences between EAPO and the equal-weight benchmark. The plot reports the top sectors by absolute weight difference, with the remainder aggregated into an ``Other sectors'' category.}
  \label{fig:sector_tilts}
\end{figure}

\vspace{0.5em}
\noindent\textbf{Benchmark tracking.}
Table~\ref{tab:tracking} reports beta, correlation, and tracking error relative to the equal-weight benchmark. EAPO maintains a high correlation with EW and a modest tracking error, consistent with a decarbonization overlay rather than a wholesale reallocation of aggregate market risk.

\input{tables/tracking_diagnostics}

\vspace{0.5em}
\noindent\textbf{Subperiod stability and inference on Sharpe differences.}
Appendix Table~\ref{tab:bootstrap_sharpe} provides block-bootstrap confidence intervals for the Sharpe difference between EAPO and EW, which are wide enough to encompass economically small positive and negative values. Appendix Table~\ref{tab:style_exposure} reports simple price-based style diagnostics computed at rebalance dates.

\vspace{0.5em}
\noindent\textbf{Turnover, frictions, and parameter sweeps.}
A natural concern is that robust decarbonization might require frequent trading and therefore high implicit and explicit costs. In our baseline configuration, average monthly $\ell_1$ turnover for EAPO is in the low single digits as a percentage of portfolio weight, so at 2\,bps per dollar traded the drag from transaction costs is well below 1\,bp per rebalance and economically negligible at a monthly horizon.

We also explore sensitivity to the risk-aversion parameter $\theta$ and to turnover caps $\tau$. Without a turnover cap, wealth trajectories for different values of $\theta$ are tightly bunched, and average Scope-1 intensities are nearly unchanged. This is economically intuitive: the $\ell_2$-based robustness term already discourages excessive concentration, so moderate changes in $\theta$ mainly smooth the weights without materially altering the emissions profile.

Imposing a turnover cap of $\tau = 0.2$ (20\% of notional per month) reduces turnover further, with little visible impact on cumulative wealth or average intensity. Pareto plots of expected return versus average Scope-1 intensity exhibit the convex, monotone trade-off predicted by Proposition~\ref{prop:pareto}: large reductions in intensity are attainable with modest movements along the expected-return axis, and the volatility ``bubble sizes'' remain comparable across the frontier rather than exploding in the low-emissions region.

\vspace{0.5em}
\noindent\textbf{Holdings transparency.}
To make the emissions reductions interpretable at the security level, Appendix Table~\ref{tab:eapo-holdings} reports the top holdings by average weight in the baseline EAPO portfolio together with their average Scope-1 intensities. Appendix Figure~\ref{fig:eapo-holdings-heatmap} visualizes the corresponding monthly weight paths. For replication and integration into reporting pipelines, we also provide a full monthly weight panel for EW, GMV, EMW, and EAPO as CSV exports in the supplementary files.

\begin{figure}[t]
  \centering
  \begin{subfigure}[t]{0.48\textwidth}
    \centering
    \includegraphics[width=\linewidth]{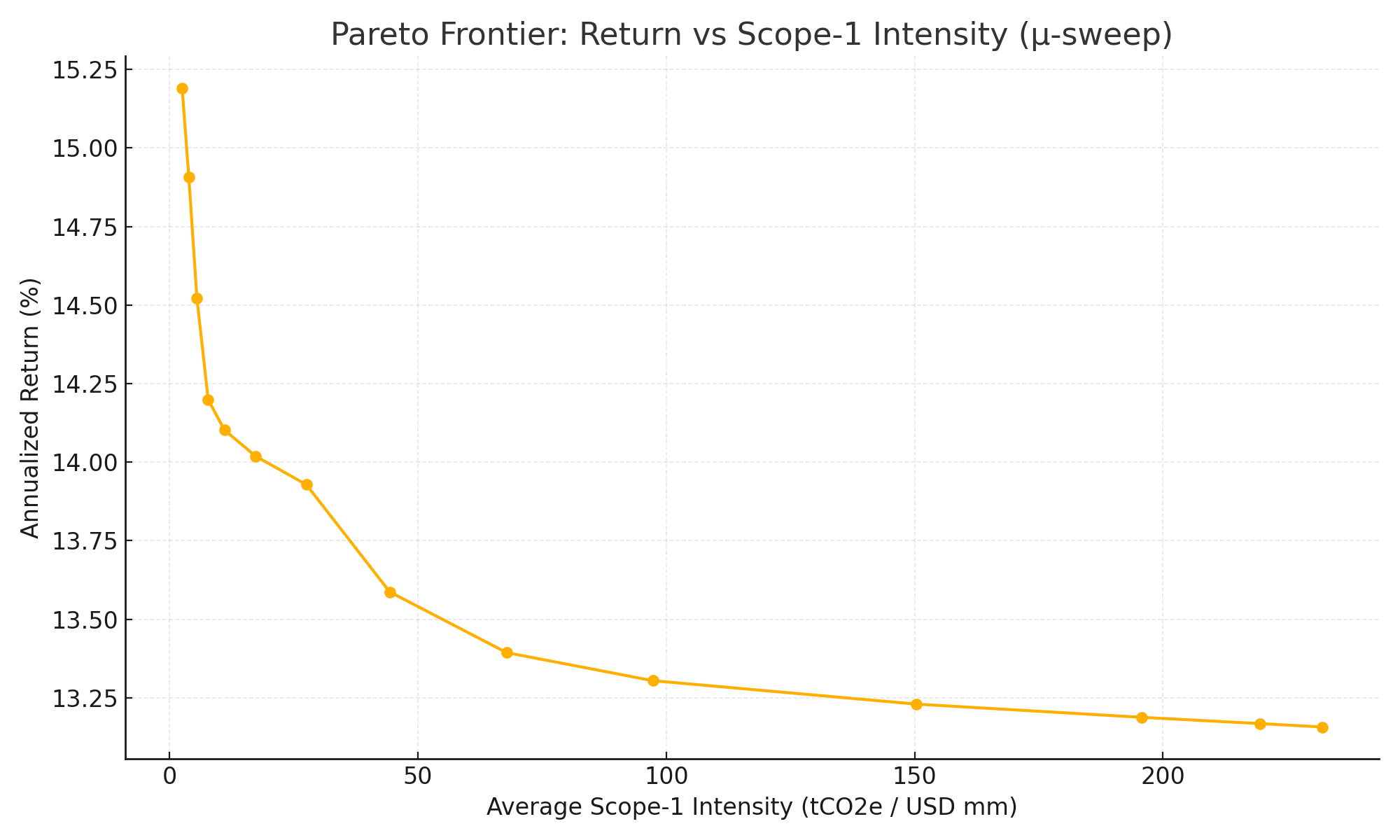}
    \caption{Pareto frontier: annualized return versus average Scope-1 intensity under a $\mu$-sweep.}
    \label{fig:pareto_frontier}
  \end{subfigure}
  \hfill
  \begin{subfigure}[t]{0.48\textwidth}
    \centering
    \includegraphics[width=\linewidth]{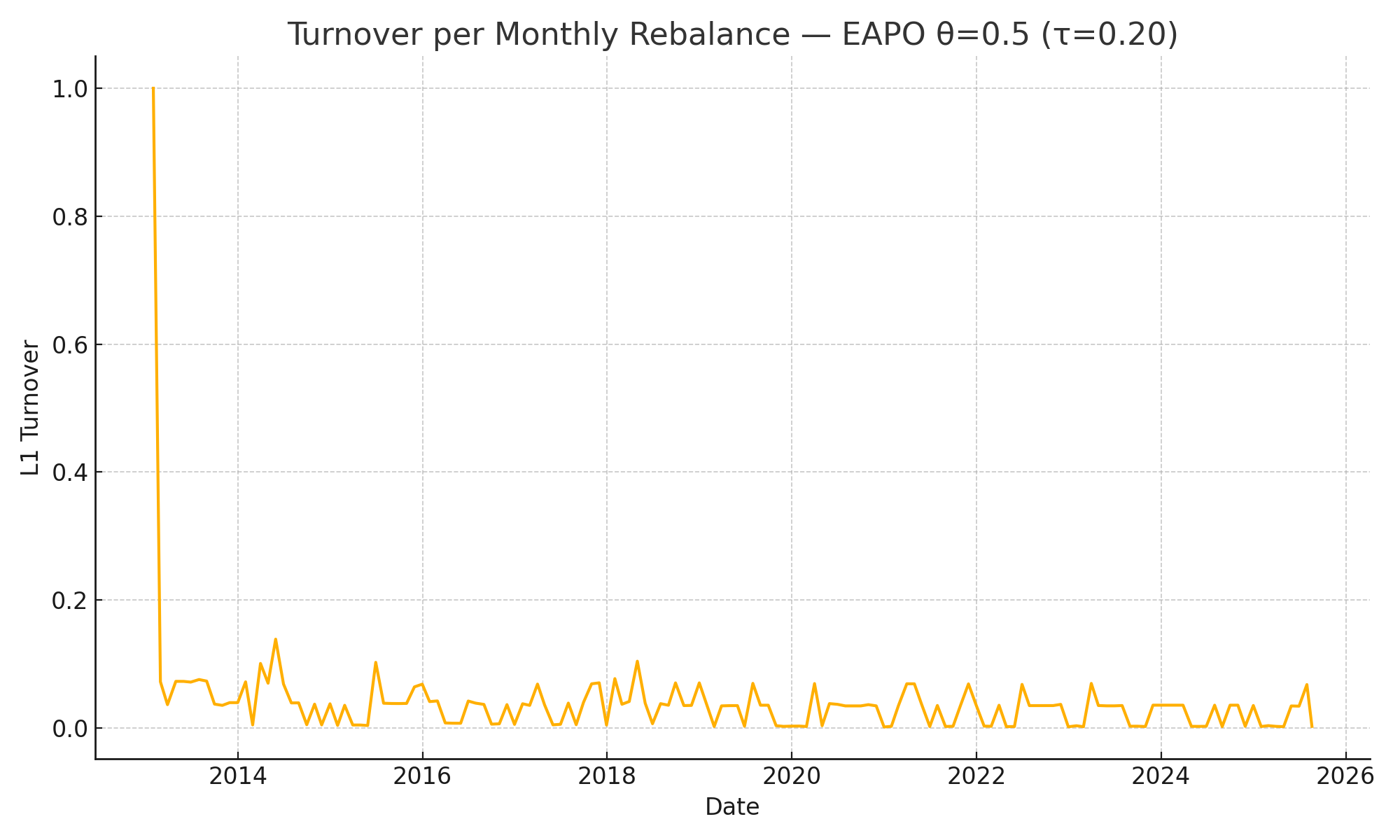}
    \caption{Monthly $\ell_1$ turnover per rebalance for EAPO ($\theta=0.5$) under a turnover cap $\tau=0.20$.}
    \label{fig:turnover_series}
  \end{subfigure}
  \caption{Return--emissions trade-off and implementability. The left panel traces the empirical return--Scope-1-intensity frontier as the scalarization weight $\mu$ varies. The right panel reports realized turnover across rebalances under a representative turnover cap.}
  \label{fig:pareto_turnover}
\end{figure}

\begin{figure}[h!]
    \centering
    \begin{subfigure}[t]{0.48\textwidth}
        \centering
        \includegraphics[width=\textwidth]{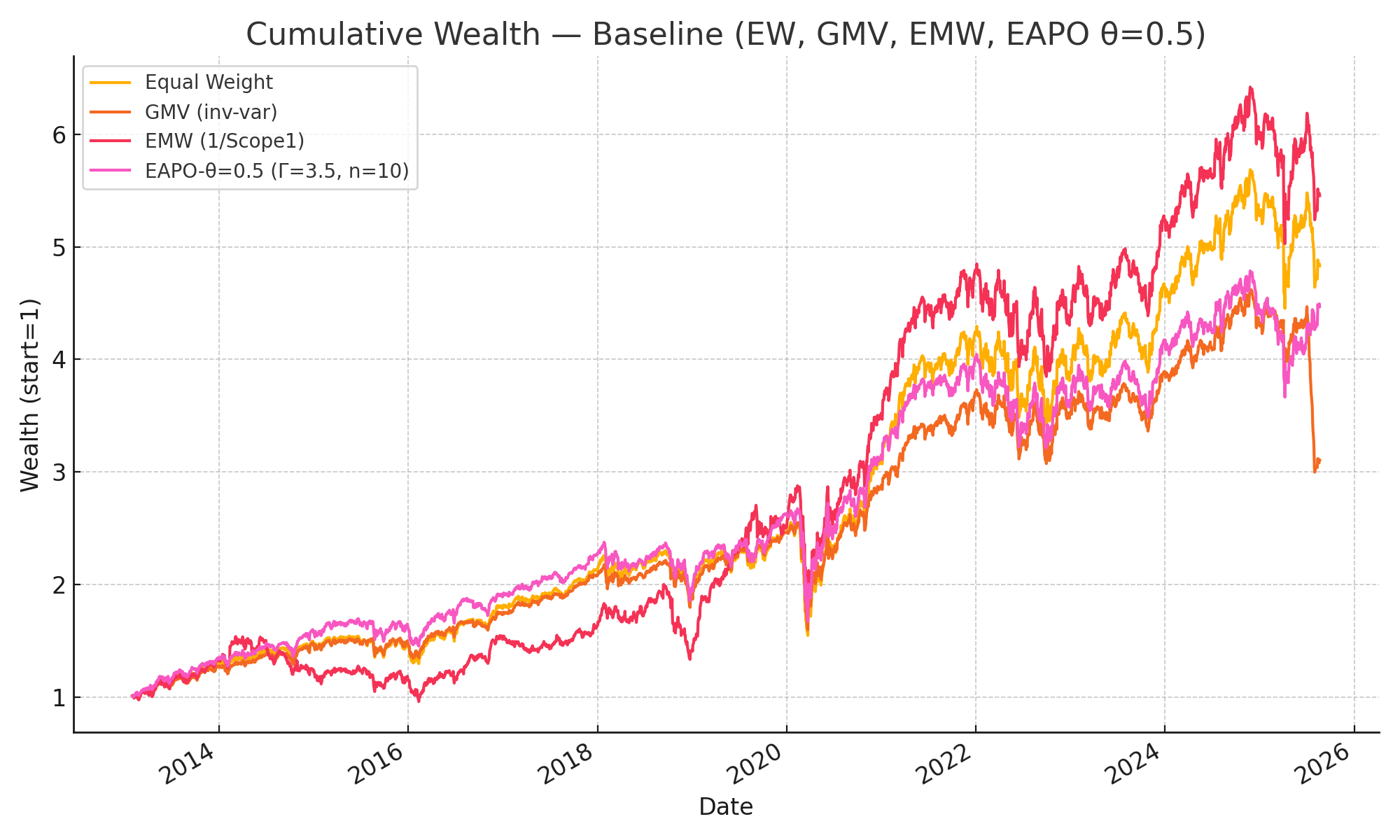}
        \caption{Cumulative wealth of EAPO, and benchmarks equal weight, global minimum variance (EWM), and inverse scope (EMW)}
    \end{subfigure}%
    \hfill
    \begin{subfigure}[t]{0.48\textwidth}
        \centering
        \includegraphics[width=\textwidth]{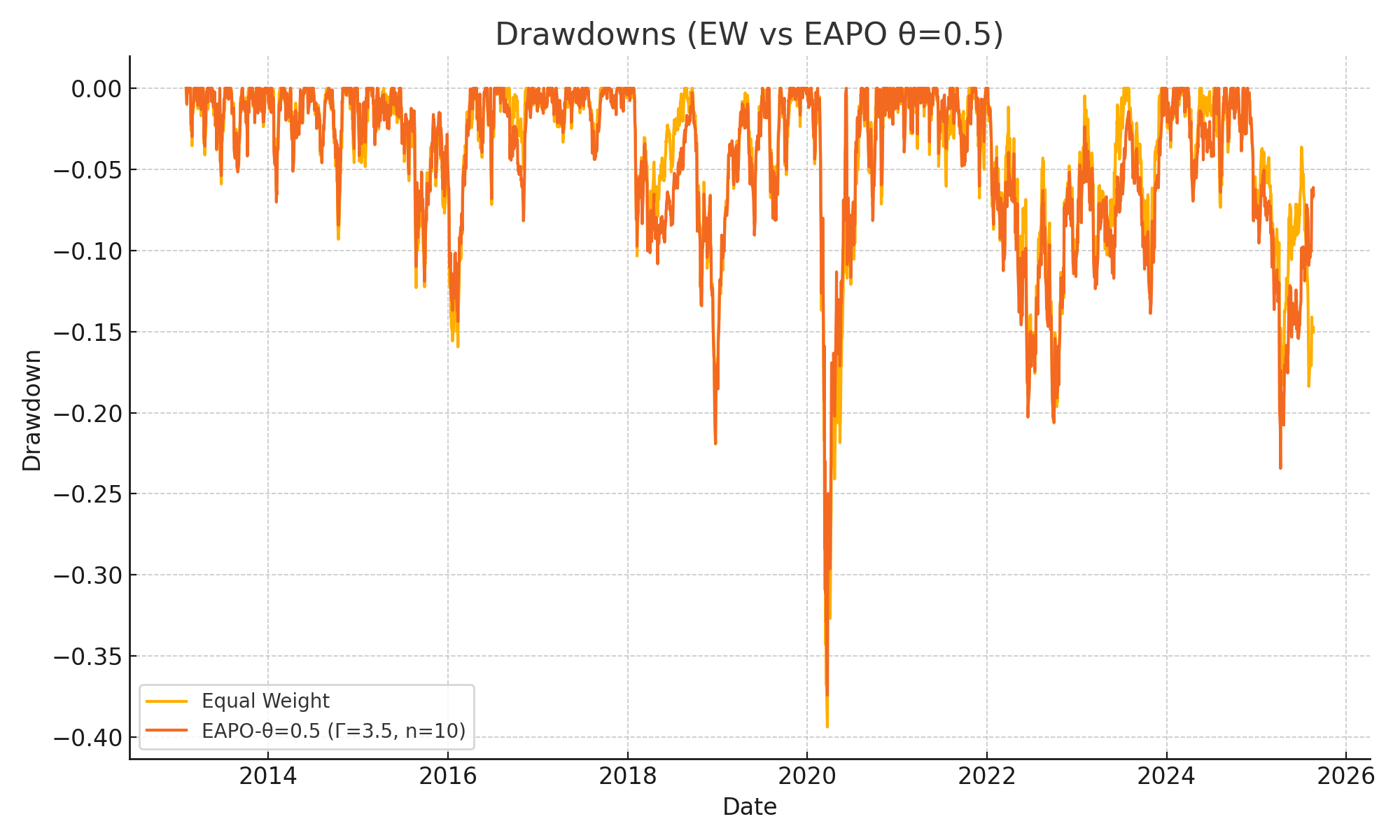}
        \caption{Max drawdown comparison between EAPO and EW}
    \end{subfigure}

    \vspace{0.4cm}

    \begin{subfigure}[t]{0.6\textwidth}
        \centering
        \includegraphics[width=\textwidth]{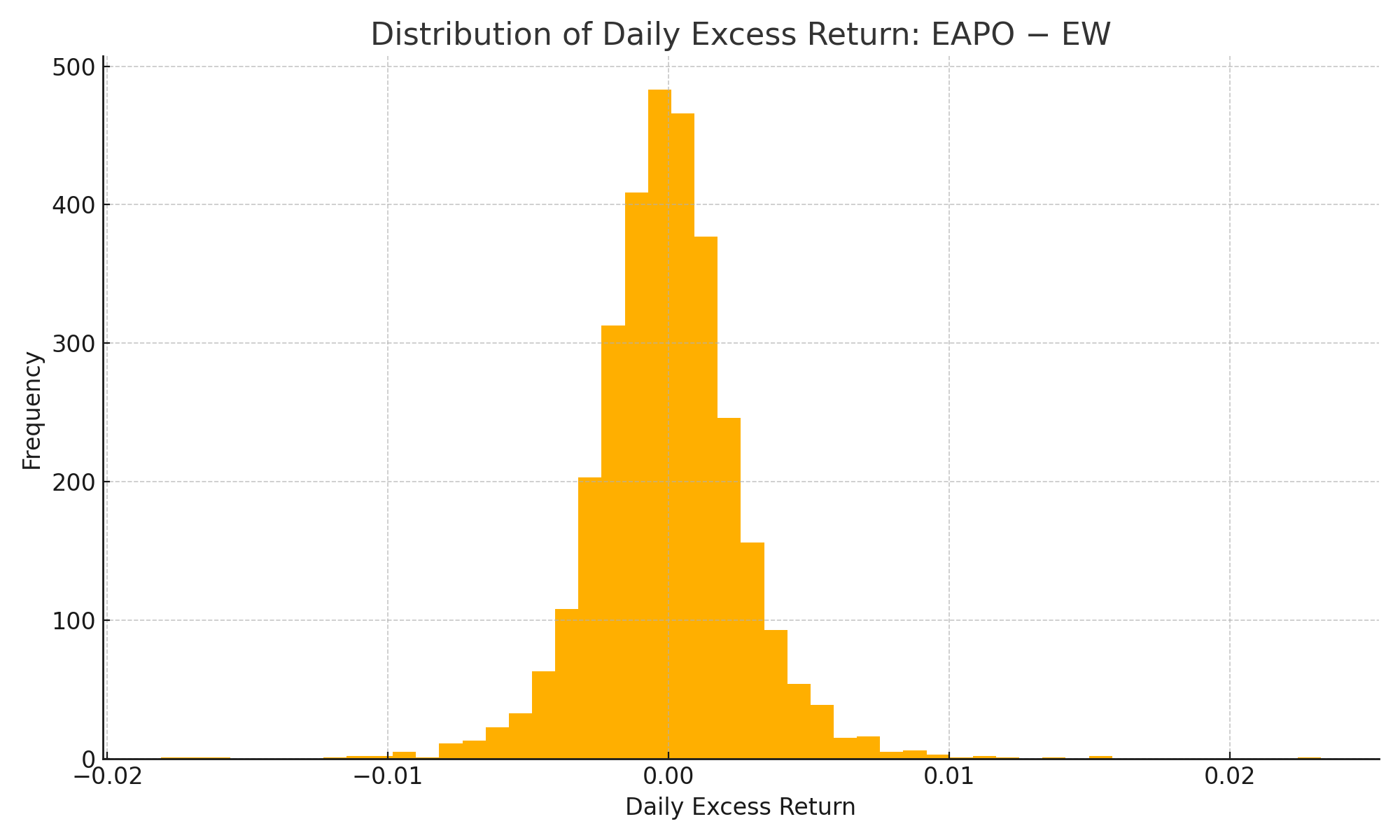}
        \caption{Distribution of daily excess return between EAPO and EW}
    \end{subfigure}

    \caption{(i)--(iii) summary of portfolio performance.}
    \label{fig:R2_portfolios}
\end{figure}

Taken together, these results support the claim that emissions-aware robustness can achieve order-of-magnitude reductions in financed Scope-1 intensity while delivering benchmark-level Sharpe ratios.

\subsection{Discussion}\label{sec:discussion-results}

Performance statistics alone do not justify deployment in an institutional setting. The strategy must also align with regulatory developments and internal governance. The empirical results above show that EAPO satisfies three requirements that matter for asset owners and regulators alike.

First, {material decarbonization}: the reductions in portfolio-level Scope-1 intensity---on the order of 80--95\% relative to standard benchmarks in our sample---are large enough to matter for financed-emissions reporting under PCAF and related frameworks. The fact that the reduction is broad-based across the book, rather than concentrated in a small sleeve, makes it easier to communicate the effect to clients and oversight committees.

Second, {financial integrity}: we do not observe statistically significant changes in average daily returns relative to EW, GMV, or EMW once serial correlation is accounted for with HAC inference. Sharpe ratios remain close to benchmark values, and any differences are largely attributable to modest changes in volatility and drawdowns rather than to pursuing climate-themed risk premia. This property is important for mandate governance: a risk committee can adopt EAPO as a sustainability overlay without implicitly accepting a fundamentally different return-generation process.

Third, {operational tractability}: the projected-gradient implementation with explicit turnover control, combined with the LP/SOCP structure of the underlying robust problem, scales comfortably to institutional universes. The algorithm respects existing long-only, fully-invested constraints, and the dual variables admit a clear interpretation as shadow carbon prices, which can be reported alongside traditional risk budgets.

\subsubsection{Why deep decarbonization can be performance neutral}
The absence of a Sharpe penalty is consistent with several mechanisms. First, emissions intensity is correlated with sector membership and with standard equity characteristics, so decarbonization can be achieved through a combination of sector allocation and within-sector selection. Table~\ref{tab:emissions_attribution} shows that, in our baseline, most of the reduction in average intensity relative to EW comes from underweighting high-intensity sectors, with a material remainder coming from reweighting toward lower-intensity firms within sectors. Figure~\ref{fig:sector_tilts} makes the associated sector tilts explicit. When sector tilts and within-sector dispersion both contribute, the opportunity cost of decarbonization depends on cross-sector risk premia and on competitive dispersion rather than on an arbitrary loss of diversification.

Second, carbon exposure can behave like a priced factor, but the associated compensation is time-varying and can be partially subsumed by other state variables. Carbon-intensive firms have earned higher average returns in some samples, consistent with a transition-risk premium \citep{Bolton_Kacperczyk_2021}. A penalty on emissions therefore reduces expected return when the carbon premium is positive. In finite samples, realized carbon premia can be weak and can co-move with standard factor returns. One plausible channel for performance neutrality is therefore a combination of moderate shadow carbon prices and correlated factor realizations over the back-test window.

Third, transition risk is asymmetric. Option-implied evidence indicates that high-emissions firms embed downside climate risk that is not well summarized by variance \citep{Ilhan_Sautner_Vilkov_2021}. By reallocating away from high-intensity firms, an emissions-aware strategy can reduce exposure to tail risk that is imperfectly priced or poorly proxied by quadratic risk. This channel can raise realized Sharpe ratios even when average returns are unchanged.

Finally, robustness changes the optimization problem in a way that can improve out-of-sample performance. The ambiguity penalty shrinks the solution away from concentrated positions that are sensitive to noisy estimates of both expected returns and emissions intensities. This regularization can lower realized volatility and turnover. It also stabilizes the emissions footprint when disclosures are revised, which reduces implementation risk for mandate reporting. In the model, the dual multiplier on the ambiguity radius provides a disciplined summary of this trade-off in return units.

In summary, the empirical evidence is consistent with the message of Sections~\ref{sec:model}--\ref{sec:empirical}: an axiomatized emissions-penalty operator, embedded in a robust mean--risk program with explicit implementability controls, can move a large long-only equity portfolio substantially down the emissions ladder while leaving conventional performance statistics essentially unchanged. This makes EAPO a credible choice for mandates in which financed emissions are a primary objective rather than a secondary disclosure metric.

\section{Conclusion and Future Research}
\label{sec:conclusion}

This paper develops and empirically validates an emissions-aware portfolio optimization (EAPO) framework that remains tractable at institutional scale while being explicitly robust to uncertainty in corporate greenhouse gas (GHG) reporting. By embedding a scope-specific, axiomatized emissions-penalty operator into convex robust programs with exact LP/SOCP representations, we show that climate objectives can be placed on the same footing as classical mean--variance and CVaR criteria. The framework preserves the familiar geometry of portfolio choice while internalizing financed emissions through a disciplined treatment of measurement error.

Our analysis delivers three main messages.  
First, on the modeling side, we characterize a unique family of curvature-controlled emissions-penalty operators that are compatible with homogeneity, scale invariance, mixture linearity, and a semigroup property. This operator nests smoothly into both moment- and distributionally-robust formulations, admits linear or second-order cone reformulations, and yields dual variables that are naturally interpretable as shadow carbon prices. These duals quantify, in return units, the marginal cost of additional robustness against misreported emissions, providing a direct bridge between optimization primitives and climate-policy narratives.

Second, on the algorithmic side, we show that emissions-aware robustness does not require sacrificing scalability. The EAPO problems we study reduce to sparse LP/SOCP instances with complexity linear in the number of assets and cones, and can be solved reliably via projected-gradient schemes on the simplex with turnover projections. This architecture is compatible with daily or monthly rebalancing, standard covariance shrinkage, and realistic transaction-cost schedules, making it implementable within existing buy-side risk and execution stacks.

Third, on the empirical side, we document that robust decarbonization is attainable at negligible performance cost in a large-cap U.S.\ equity universe. In our baseline S\&P~500 experiment with monthly rebalancing, EAPO reduces average Scope~1 portfolio emissions intensity by roughly an order of magnitude relative to the equal-weight benchmark, and by more than 80\% relative to a simple emissions-weighted portfolio, while delivering Sharpe ratios that are statistically indistinguishable from those benchmarks under block-bootstrap inference. Drawdowns remain comparable or slightly attenuated, transaction-cost drag is modest under realistic turnover caps, and the return--emissions Pareto frontier exhibits the predicted convex shape: large reductions in financed emissions are available at small marginal costs in risk-adjusted performance. Taken together, these results suggest that properly designed robust optimization can reconcile sustainability constraints with traditional portfolio objectives in a way that is both economically meaningful and operationally realistic.

\subsection{Managerial and Policy Implications}

For practitioners, the EAPO framework is best viewed as an overlay technology rather than a wholesale replacement for existing investment processes. Because the emissions operator acts directly on the return kernel, it can be composed with standard constraints (e.g., sector, factor, liquidity) and integrated into existing mean--variance, risk-parity, or factor-model architectures with minimal re-engineering.

From an asset-management perspective, three implications are immediate.  
(i) {Implementability at scale.} The conic reformulations and first-order algorithms in Sections~\ref{sec:model} and~\ref{sec:results} imply that emissions-robust overlays can be deployed on institutional universes (hundreds to thousands of names) under standard end-of-day time budgets. This makes EAPO compatible with large index-plus and enhanced passive mandates, where operational frictions are often the binding constraint.  
(ii) {Governance via shadow prices.} The dual variables associated with the ambiguity sets deliver an internal ``shadow carbon tax'' that measures the marginal Sharpe reduction per unit tightening of the robustness budget~$\Gamma$. This provides investment committees and risk officers with a quantitative, dollar-denominated knob for calibrating how aggressively to trade off expected return against protection from emissions misreporting, and for documenting those trade-offs to clients and boards.  
(iii) {Robust reporting of financed emissions.} Because the framework explicitly models estimation error in $\lambda$ and propagates it into portfolio-level exposures, it yields intensity and ``emissions yield'' metrics (emissions per unit of excess return) that can be reported alongside traditional performance statistics. This is directly aligned with emerging disclosure regimes under IFRS~S2, PCAF, and related climate-reporting standards, and can reduce the probability that portfolios appear green in-sample but fail to decarbonize once misreporting is corrected.

For regulators and benchmark providers, the main implication is methodological: EAPO-type constructions demonstrate that Paris-aligned and climate-transition benchmarks need not rely solely on hard exclusions or ad hoc tilts. Instead, robust, optimization-based designs can internalize both transition pathways and disclosure noise, producing indices and supervisory stress tests that are transparent, reproducible, and less sensitive to idiosyncrasies in vendor data. Finally, for corporate issuers, the framework highlights a concrete channel through which higher-quality emissions disclosure can relax investors' robustness penalties, lowering perceived risk premia and potentially improving access to capital.

\subsection{Open Problems and Future Research}

Several limitations of the present analysis point to a research agenda at the intersection of finance, optimization, and climate science.

First, our empirical work focuses on Scope~1 emissions in a single large-cap U.S.\ equity universe. Extending the framework to multi-scope, multi-region settings---including Scope~2 and material Scope~3 categories, as well as financed-emissions metrics for credit and private assets---would allow a more complete mapping between portfolio construction and real-economy decarbonization. Doing so will require both richer data and new modeling for the dependence structure between scopes and across corporate value chains.

Second, the dynamic analysis in Section~\ref{sec:model} treats the investor as small relative to the market and abstracts from general-equilibrium feedbacks. As robust, emissions-aware strategies scale, they will affect prices, cost of capital, and firms' incentives to disclose and abate. Embedding EAPO-style preferences into equilibrium asset-pricing models, and quantifying the induced changes in emissions trajectories, remains a central open problem with direct relevance for sustainable-finance regulation.

Third, while we incorporate transition and disclosure risk through robust drift and intensity terms, we work with relatively standard return-generating processes and linear penalties. A natural extension is to couple EAPO with heavy-tailed or jump processes that better capture physical and transition tail events, as well as with forward-looking scenario sets such as NGFS pathways. This would allow a unified treatment of chronic and acute climate risk within the same robust optimization backbone.

Fourth, our empirical implementation uses a single class of ambiguity sets and a particular choice of curvature parameter~$m$. Although the sensitivity and Lipschitz results in Sections~\ref{sec:model} and~\ref{sec:results} show that performance deteriorates sub-linearly as robustness is tightened, a more systematic comparison of alternative divergence measures, ambiguity geometries, and learning rules for~$\Gamma$ is warranted. In particular, data-driven updating of robustness budgets---for example, as disclosure quality or verification coverage improves---could make the framework adaptive over the life of a mandate.

Finally, there is an open methodological opportunity on the interface with corporate and operational decision-making. On the real-economy side, firm-level abatement projects, capital-budgeting decisions, and supply-chain redesigns are increasingly evaluated under internal carbon prices and scenario analysis. On the financial side, investors are beginning to use shadow carbon prices and emissions-adjusted discount rates. Developing joint models that couple EAPO-type portfolio optimization with firm-level real-options or capacity-planning problems would move the literature closer to a fully integrated view of climate risk, capital allocation, and operations.

In summary, our results indicate that robust optimization provides a viable and implementable pathway for reconciling the dual imperatives of portfolio performance and climate responsibility. The combination of axiomatic modeling, tractable convex reformulations, and statistically disciplined empirical evidence suggests that emissions-robust portfolio design can be a core tool---rather than an add-on---in the next generation of sustainable investment mandates.

\bibliographystyle{plainnat}
\bibliography{references}

\section{Appendix}

\subsection{Axioms and uniqueness of the emissions--penalty operator}
\label{sec:appendix-penalty-axioms}

Fix a scope $j$ and let $\lambda_{i,j} \ge 0$ denote the emissions intensity of asset $i$ at the decision date,
with
\[
  \lambda_{\max,j} \equiv \max_{i} \lambda_{i,j} \in (0,\infty).
\]
For a curvature parameter $m \in \mathbb{N}_{+}$, an emissions--penalty operator is a mapping
\[
\begin{aligned}
P^{(m)}_j
    &:\,
    \mathbb{R} \times [0,\lambda_{\max,j}]
    \,\to\,
    \mathbb{R},
\\[4pt]
(r,\lambda)
    &\mapsto
    \text{``emissions-adjusted'' payoff}.
\end{aligned}
\]
The following hold for {generator} $P^{(1)}_j$.

\paragraph{Axiom (H) --- payoff homogeneity.}
For all $\alpha \ge 0$ and $r \in \mathbb{R}$,
\[
  P^{(1)}_j(\alpha r,\lambda) = \alpha\, P^{(1)}_j(r,\lambda).
\]

\paragraph{Axiom (N) --- normalization.}
For all $r \in \mathbb{R}$,
\[
  P^{(1)}_j(r,0) = r,
  \qquad
  P^{(1)}_j(r,\lambda_{\max,j}) = 0.
\]

\paragraph{Axiom (M) --- monotonicity in intensity.}
If $0 \le \lambda_1 \le \lambda_2 \le \lambda_{\max,j}$ and $r \ge 0$, then
\[
  P^{(1)}_j(r,\lambda_1) \ge P^{(1)}_j(r,\lambda_2).
\]

\paragraph{Axiom (SI) --- scale invariance of units.}
For any $\beta > 0$,
\[
  P^{(1)}_j(r,\lambda,\lambda_{\max,j})
  =
  P^{(1)}_j(r,\beta \lambda, \beta \lambda_{\max,j}),
\]
that is, rescaling the emissions unit leaves the transformed payoff unchanged.

\paragraph{Axiom (L) --- mixture linearity in intensity.}
For all $\theta \in [0,1]$ and $\lambda_1,\lambda_2 \in [0,\lambda_{\max,j}]$,
\[
\begin{aligned}
P^{(1)}_j\!\bigl(r, \theta\lambda_1 + (1-\theta)\lambda_2\bigr)
    &=
    \theta\, P^{(1)}_j(r,\lambda_1)
\\[4pt]
    &\quad
    + (1-\theta)\, P^{(1)}_j(r,\lambda_2).
\end{aligned}
\]

\paragraph{Axiom (C) --- curvature semigroup.}
There exists a family $\{P^{(m)}_j\}_{m\in\mathbb{N}_{+}}$ with generator $P^{(1)}_j$ such that
\[
  P^{(m_1+m_2)}_j
  =
  P^{(m_1)}_j \circ P^{(m_2)}_j, \qquad \forall\, m_1,m_2 \in \mathbb{N}_{+}.
\]

We additionally assume that $P^{(1)}_j$ is continuous in $\lambda$.

\begin{proposition}[Uniqueness and closed form]
\label{prop:uniqueness-penalty}
Under axioms {\rm(H)--(N)--(M)--(SI)--(L)--(C)} and continuity in $\lambda$, the only admissible family
$\{P^{(m)}_j\}_{m\in\mathbb{N}_{+}}$ is
\begin{equation}
\label{eq:penalty-closed-form}
\begin{aligned}
P^{(m)}_j(r,\lambda)
    &=
    \Bigl(1 - \frac{\lambda}{\lambda_{\max,j}}\Bigr)^m r,
\\[6pt]
(r,\lambda)
    &\in \mathbb{R} \times [0,\lambda_{\max,j}] .
\end{aligned}
\end{equation}
\end{proposition}

\begin{proof}
\ We give a fully explicit argument that isolates where each axiom enters.
As a compact summary, payoff homogeneity (H) and unit-scale invariance (SI) imply the existence of a
scalar function $\varphi^{(1)}_j:[0,1]\to[0,1]$ such that
$P^{(1)}_j(r,\lambda)=\varphi^{(1)}_j(\lambda/\lambda_{\max,j})\,r$.
The steps below justify this representation and pin down $\varphi^{(1)}_j$ uniquely.

Fix $r \in \mathbb{R}$ and define the one-dimensional section
\[
  f_r(x)
  :=
  P^{(1)}_j\bigl(r, x\,\lambda_{\max,j}\bigr),
  \qquad x \in [0,1].
\]
Unit-scale invariance (SI) implies that the dependence on $(\lambda,\lambda_{\max,j})$ is only through
the ratio $x=\lambda/\lambda_{\max,j}$, hence
\begin{equation}
  P^{(1)}_j(r,\lambda) = f_r\!\Bigl(\frac{\lambda}{\lambda_{\max,j}}\Bigr).
  \label{eq:appendix-fr-reduction}
\end{equation}

By mixture linearity (L), for any $x_1,x_2 \in [0,1]$ and $\theta \in [0,1]$,
\[
\begin{aligned}
f_r\bigl(\theta x_1 + (1-\theta)x_2\bigr)
&=
P^{(1)}_j\bigl(r,(\theta x_1 + (1-\theta)x_2)\lambda_{\max,j}\bigr)
\\[3pt]
&=
\theta\,P^{(1)}_j\bigl(r,x_1\lambda_{\max,j}\bigr)
 + (1-\theta)\,P^{(1)}_j\bigl(r,x_2\lambda_{\max,j}\bigr)
\\[3pt]
&=
\theta f_r(x_1) + (1-\theta)f_r(x_2).
\end{aligned}
\]
Thus $f_r$ satisfies Jensen's functional equation on the interval $[0,1]$.
Since $P^{(1)}_j$ is continuous in $\lambda$ by assumption, $f_r$ is continuous in $x$.
The standard regularity result for Jensen's equation on an interval then yields that $f_r$ is affine.
Concretely, there exist constants $a(r),b(r) \in \mathbb{R}$ such that
\begin{equation}
  f_r(x) = a(r)\,x + b(r), \qquad x \in [0,1].
  \label{eq:appendix-fr-affine}
\end{equation}

Normalization (N) implies
\[
  f_r(0)=P^{(1)}_j(r,0)=r,
  \qquad
  f_r(1)=P^{(1)}_j(r,\lambda_{\max,j})=0.
\]
Substituting into \eqref{eq:appendix-fr-affine} gives $b(r)=r$ and $a(r)=-r$. Therefore
\[
  f_r(x) = r(1-x), \qquad x\in[0,1].
\]
Combining with \eqref{eq:appendix-fr-reduction} yields the closed form
\begin{equation}
  P^{(1)}_j(r,\lambda)
  =
  \Bigl(1-\frac{\lambda}{\lambda_{\max,j}}\Bigr)\,r,
  \qquad (r,\lambda)\in\mathbb{R}\times[0,\lambda_{\max,j}].
  \label{eq:appendix-p1-closed}
\end{equation}

Define $\varphi^{(1)}_j:[0,1]\to[0,1]$ by $\varphi^{(1)}_j(x):=1-x$. Then \eqref{eq:appendix-p1-closed}
is equivalently
\begin{equation}
  P^{(1)}_j(r,\lambda)
  =
  \varphi^{(1)}_j\!\Big(\frac{\lambda}{\lambda_{\max,j}}\Big)\, r.
  \label{eq:phi-def}
\end{equation}
Payoff homogeneity (H) is consistent with this representation and will be inherited by the full curvature
family.
Normalization (N) implies
\[
  \varphi^{(1)}_j(0) = 1,
  \qquad
  \varphi^{(1)}_j(1) = 0,
\]
and monotonicity (M) guarantees that $\varphi^{(1)}_j$ is nonincreasing.  

Mixture linearity (L), combined with \eqref{eq:phi-def}, gives for all
$x_1,x_2 \in [0,1]$ and $\theta \in [0,1]$,
\[
\begin{aligned}
\varphi^{(1)}_j\!\bigl(\theta x_1 + (1-\theta)x_2\bigr)
    &=
    \theta\, \varphi^{(1)}_j(x_1)
\\[4pt]
    &\quad
    + (1-\theta)\, \varphi^{(1)}_j(x_2).
\end{aligned}
\]
This is Jensen's functional equation on $[0,1]$. Under continuity, the only solutions are affine functions.
One way to see this is to first use the midpoint identity
$\varphi^{(1)}_j\bigl(\tfrac{x+y}{2}\bigr)=\tfrac{\varphi^{(1)}_j(x)+\varphi^{(1)}_j(y)}{2}$,
iterate to obtain the identity for all dyadic rationals, extend to all rationals by induction on denominators,
and then use continuity to extend to all real $x\in[0,1]$. Thus there exist $a_j,b_j \in \mathbb{R}$ such that
\[
  \varphi^{(1)}_j(x) = a_j x + b_j, \qquad x \in [0,1].
\]
Using $\varphi^{(1)}_j(0)=1$ and $\varphi^{(1)}_j(1)=0$ gives $b_j = 1$ and $a_j = -1$, hence
\[
\begin{aligned}
\varphi^{(1)}_j(x) &= 1 - x
\\[4pt]
&\Longrightarrow\quad
P^{(1)}_j(r,\lambda)
    = \Bigl(1 - \frac{\lambda}{\lambda_{\max,j}}\Bigr) r.
\end{aligned}
\]

Axiom (C) now determines all higher curvatures. For any $m\in\mathbb{N}_+$, the semigroup identity implies
\[
  P^{(m)}_j
  =
  \underbrace{P^{(1)}_j\circ P^{(1)}_j\circ\cdots\circ P^{(1)}_j}_{m \text{ compositions}},
\]
which is proved by a short induction on $m$ using
$P^{(m+1)}_j = P^{(m)}_j\circ P^{(1)}_j$.
Substituting \eqref{eq:appendix-p1-closed} into this composition gives
\[
  P^{(m)}_j(r,\lambda)
  =
  \Big(1 - \frac{\lambda}{\lambda_{\max,j}}\Big)^m r.
\]
This family clearly satisfies the semigroup relation and inherits the remaining axioms by construction.
Conversely, if another family satisfies (C) with the same generator $P^{(1)}_j$, the same induction shows
it must coincide with the $m$-fold composition above, hence the representation \eqref{eq:penalty-closed-form}
is unique.
\end{proof}

Proposition~\ref{prop:uniqueness-penalty} shows that the curvature family used is not an
ad hoc choice: it is the unique continuous, scale-invariant, and mixture-linear way to modulate payoffs by
relative emissions intensity while preserving a semigroup structure over $m$.

\subsection{Dynamic robust Bellman recursion (Proposition~\texorpdfstring{2.2}{2.2})}
\label{sec:appendix-bellman}

For completeness we reproduce the dynamic value function. Let $t=0,\dots,T$ index decision dates, and let
$\beta \in (0,1]$ denote the discount factor. Given an initial portfolio $x_t \in \Delta^{n-1}$, define
\[
\begin{aligned}
V_t(x_t)
    &:=
    \max_{x_t,\dots,x_T \in \Delta^{\,n-1}}
\\[4pt]
    &\qquad
    \min_{z_s \in Z_s,\, s=t,\dots,T}
\\[4pt]
    &\qquad
    \sum_{s=t}^T
    \beta^{\,s-t}\,
    x_s^\top\bigl(\gamma_s + \Delta_s z_s\bigr).
\end{aligned}
\]
with $V_{T+1}(\cdot) \equiv 0$.

\paragraph{Remark (well-posedness and time consistency).}
The display above follows the notation in the main text and writes $V_t(x_t)$ while also maximizing
over $x_t$. One may interpret $V_t$ as the pre-decision value at time $t$, in which case the argument
$x_t$ is only a mnemonic for the time-$t$ control. Alternatively, one can define a post-decision value
$\widetilde V_t(\bar x_t)$ that conditions on a pre-existing holding $\bar x_t$ and includes transaction
costs or turnover constraints in the transition. In either convention, the Bellman recursion below is
valid under the standard rectangularity condition that the uncertainty set factorizes as a Cartesian
product $Z_t\times Z_{t+1}\times\cdots\times Z_T$, as written in the inner minimization.

\begin{proof}
\ We proceed by backward induction on $t$.

\paragraph{Base case $t=T$.}
At the terminal date $t=T$ we have
\[
  V_T(x_T)
  =
  \max_{x_T \in \Delta^{n-1}}
  \min_{z_T \in Z_T} x_T^\top(\gamma_T + \Delta_T z_T),
\]
which matches the Bellman form with $V_{T+1} \equiv 0$.

\paragraph{Induction step.}
Assume now that the recursion holds at $t+1$. Starting from date $t$, with 
\[
\begin{aligned}
V_{t+1}(x_{t+1})
&=
\max_{x_{t+1},\dots,x_T \in \Delta^{\,n-1}}
\min_{z_{t+1:T} \in Z_{t+1:T}}
\\[4pt]
&\qquad
\sum_{s=t+1}^T
\beta^{\,s-(t+1)}\,
x_s^\top(\gamma_s + \Delta_s z_s).
\end{aligned}
\]

Because the uncertainty is rectangular, the inner minimization over $z_t,\dots,z_T$ can be written as
an iterated minimization. For any fixed control sequence $\{x_s\}_{s=t}^T$, we have
\[
\begin{aligned}
\min_{z_s\in Z_s,\, s=t,\dots,T}
\sum_{s=t}^T \beta^{s-t} x_s^\top(\gamma_s+\Delta_s z_s)
&=
\min_{z_t\in Z_t}
\Bigl\{
x_t^\top(\gamma_t+\Delta_t z_t)
\\[2pt]
&\qquad
+\beta\,\min_{z_{t+1:T}\in Z_{t+1:T}}
\sum_{s=t+1}^T \beta^{s-(t+1)} x_s^\top(\gamma_s+\Delta_s z_s)
\Bigr\}.
\end{aligned}
\]
This step is purely algebraic and uses that $z_t$ appears only in the period-$t$ term.

\begin{equation}
\label{eq:Bellman-recursion}
\begin{aligned}
V_t(x_t)
    &=
    \max_{x_t \in \Delta^{\,n-1}}
    \min_{z_t \in Z_t}
    \Bigl\{
        x_t^\top(\gamma_t + \Delta_t z_t)
\\
    &\qquad
        + \beta        
            \max_{x_{t+1},\dots,x_T \in \Delta^{\,n-1}}
            \min_{z_{t+1:T} \in Z_{t+1:T}}  \\
            &\qquad \sum_{s=t+1}^T
                \beta^{\,s-(t+1)}
                x_s^\top(\gamma_s + \Delta_s z_s)        
    \Bigr\}.
\end{aligned}
\end{equation}

The term in braces depends on the future controls only through $x_{t+1}$, so
\[
\begin{aligned}
V_t(x_t)
&=
\max_{x_t \in \Delta^{\,n-1}}
\min_{z_t \in Z_t}
\Bigl\{
x_t^\top(\gamma_t + \Delta_t z_t)
\\
&\qquad
+ \beta\, V_{t+1}(x_{t+1})
\Bigr\}.
\end{aligned}
\]
Compactness of $Z_t$ and continuity of $z\mapsto x_t^\top(\gamma_t+\Delta_t z)$ ensure existence of a
minimizer $z_t^\star(x_t)$ for each feasible $x_t$. Concavity of the outer maximization over the simplex
follows from linearity of the stage payoff in $x_t$. This completes the backward recursion.
\end{proof}

\subsection{\texorpdfstring{$\varphi$}{phi}--divergence dual reformulation (Theorem~\texorpdfstring{2.2}{2.2})}
\label{sec:appendix-phi-dro}

Consider the distributionally robust problem
\[
\begin{aligned}
\max_{x \in \Delta^{\,n-1}}
\,\inf_{P \in \mathcal{P}}
    \,\mathbb{E}_{z\sim P}\bigl[\ell_x(z)\bigr],
\\[6pt]
\ell_x(z)
    := \sum_{i=1}^n x_i\bigl(\gamma_i + \delta_i z\bigr).
\end{aligned}
\]
with ambiguity set
\[
  \mathcal{P}
  :=
  \big\{P : D_{\varphi}(P \,\|\, \widehat{P}) \le \theta\big\},
\]
where $D_{\varphi}$ is a $\varphi$--divergence centered at the empirical law $\widehat{P}$. For clarity, suppose $\widehat{P}$ is discrete with support $\{z^{(m)}\}_{m=1}^M$ and probabilities $\hat{p}_m > 0$. Any $P \ll \widehat{P}$ admits a Radon--Nikodym density $w = \frac{dP}{d\widehat{P}}$
satisfying $w(z^{(m)}) \ge 0$ and $\sum_{m} w(z^{(m)})\hat{p}_m = 1$. The robust inner problem can then be
written as
\[
\begin{aligned}
\inf_{w}\quad
&\sum_{m=1}^M \ell_x\bigl(z^{(m)}\bigr)\, w_m \hat{p}_m
\\[6pt]
\text{s.t.}\quad
&\sum_{m=1}^M \varphi(w_m)\,\hat{p}_m \,\le\, \theta,
\\
&\sum_{m=1}^M w_m \hat{p}_m \,=\, 1,
\\
& w_m \ge 0,\quad m=1,\dots,M.
\end{aligned}
\]

\begin{proof}
 
The feasible set is nonempty since $w_m \equiv 1$ satisfies $\sum_m w_m\hat p_m = 1$ and
$\sum_m \varphi(1)\hat p_m = \varphi(1)$. For standard $\varphi$--divergences one has $\varphi(1)=0$, hence
for any $\theta>0$ this choice is strictly feasible for the divergence inequality. This Slater point
ensures strong duality for the convex program in $w$.

Introduce Lagrange multipliers $\lambda \ge 0$ for the divergence constraint and $\nu \in \mathbb{R}$ for the
normalization constraint. The Lagrangian of the inner problem is
\[
\begin{aligned}
L(w,\lambda,\nu)
    &=
    \sum_{m=1}^M \ell_x\bigl(z^{(m)}\bigr)\, w_m \hat{p}_m
\\[4pt]
    &\quad
    +\, \lambda
      \Bigl(
        \sum_{m=1}^M \varphi(w_m)\,\hat{p}_m
        - \theta
      \Bigr)
\\[4pt]
    &\quad
    +\, \nu
      \Bigl(
        \sum_{m=1}^M w_m \hat{p}_m
        - 1
      \Bigr).
\end{aligned}
\]
Rearranging terms yields
\[
\begin{aligned}
L(w,\lambda,\nu)
    &=
    -\theta\,\lambda
    - \nu
\\[6pt]
    &\qquad
    + \sum_{m=1}^M \hat{p}_m
      \Bigl[
          \ell_x\bigl(z^{(m)}\bigr)\, w_m
\\
    &\qquad\qquad\qquad
          + \lambda\, \varphi(w_m)
\\
    &\qquad\qquad\qquad
          + \nu\, w_m
      \Bigr].
\end{aligned}
\]

Because the constraints are separable across $m$, the Lagrangian decouples across the coordinates
$w_m$. Write the coefficient on $w_m$ as $c_m := \ell_x\bigl(z^{(m)}\bigr)+\nu$. For fixed $(\lambda,\nu)$
we therefore need the scalar infimum
\[
  \inf_{u\ge 0} \{c_m u + \lambda\varphi(u)\}.
\]
Recall the convex conjugate $\varphi^\star(y) := \sup_{u\ge 0}\{u y-\varphi(u)\}$. A standard identity
then gives
\[
  \inf_{u\ge 0}\{\varphi(u)+y u\} = -\varphi^\star(-y).
\]
Applying this with $y=c_m/\lambda$ when $\lambda>0$ yields the claimed expression. To match the more
common DRO parametrization in which the shift variable appears with the sign $(\ell-\nu)/\lambda$,
one may equivalently substitute $\nu\leftarrow -\nu$ at the end of the derivation.

For fixed $(\lambda,\nu)$, the pointwise infimum over $w_m \ge 0$ is governed by the convex conjugate
$\varphi^\star$:
\[
\begin{aligned}
\inf_{w_m \ge 0}
&\bigl\{
    \ell_x\bigl(z^{(m)}\bigr)\, w_m
    + \lambda \varphi(w_m)
    + \nu w_m
 \bigr\}
\\[4pt]
&=
-\lambda\,
\varphi^\star\!\Bigl(
    \frac{\ell_x\bigl(z^{(m)}\bigr) - \nu}{\lambda}
\Bigr).
\end{aligned}
\]
with the convention that $\lambda \varphi^\star(\cdot)$ is $+\infty$ when $\lambda = 0$ and the argument lies
outside the effective domain. Substituting back,
\[
\begin{aligned}
\inf_{w} L(w,\lambda,\nu)
&=
\nu - \theta\,\lambda
\\[6pt]
&\qquad
- \lambda
  \sum_{m=1}^M \hat{p}_m\,
  \varphi^\star\!\Bigl(
      \frac{\ell_x(z^{(m)}) - \nu}{\lambda}
  \Bigr).
\end{aligned}
\]
Taking the supremum over $(\lambda,\nu)$ produces the dual representation
\[
\begin{aligned}
\inf_{P \in \mathcal{P}} \mathbb{E}_{P}[\ell_x(z)]
&=
\sup_{\lambda \ge 0,\,\nu \in \mathbb{R}}
\Biggl\{
    \nu
\\[4pt]
&\qquad
    - \theta\,\lambda
\\[4pt]
&\qquad
    - \lambda \sum_{m=1}^M \hat{p}_m\,
      \varphi^\star\!\Bigl(
          \frac{\ell_x(z^{(m)}) - \nu}{\lambda}
      \Bigr)
\Biggr\}.
\end{aligned}
\]
Strong duality holds under standard regularity conditions (e.g., Slater's condition for the primal), so the
robust problem is equivalent to maximizing this dual objective over $x \in \Delta^{n-1}$. This is precisely
the scalar dual formulation stated.
\end{proof}

\subsection{Continuous \texorpdfstring{$\varphi$}{phi}--DRO reformulation}
\label{sec:appendix-phi-continuous}

The preceding argument extends verbatim beyond a discrete empirical support. For a general ambiguity set
\[
  \mathcal{P}
  =
  \big\{P : D_{\varphi}(P \,\|\, \widehat{P}) \le \theta\big\}
\]
and linear loss $\ell_x$, Fenchel duality yields
\begin{equation}
\label{eq:phi-general}
\begin{aligned}
\sup_{x \in \Delta^{\,n-1}}
\inf_{P \in \mathcal{P}} \mathbb{E}_{P}[\ell_x(z)]
&=
\sup_{x \in \Delta^{\,n-1}}
\sup_{\lambda \ge 0,\,\nu \in \mathbb{R}}
\Biggl\{
    \nu
\\
&\qquad
    - \theta\,\lambda
\\
&\qquad
    - \lambda\,
      \mathbb{E}_{\widehat{P}}
      \Bigl[
        \varphi^\star\!\Bigl(
            \tfrac{\ell_x(z) - \nu}{\lambda}
        \Bigr)
      \Bigr]
\Biggr\}.
\end{aligned}
\end{equation}
The distributional uncertainty enters only through the scalar pair $(\lambda,\nu)$ and the expectation of
$\varphi^\star$, so the robustification adds two decision variables and a one-dimensional nonlinearity
independently of the dimension of $z$.

\subsection{Convexity of the return--emissions Pareto frontier (Proposition~\texorpdfstring{2.3}{2.3})}
\label{sec:appendix-pareto}

Consider the scalarized problem
\[
  \max_{x \in \Delta^{n-1}}
  F_\mu(x)
  :=
  x^\top r - \mu\, x^\top \lambda,
  \qquad \mu \ge 0,
\]
and denote the optimal value by
\[
  g(\mu) := \max_{x \in \Delta^{n-1}} F_\mu(x).
\]
The attainable set of pairs $(x^\top r, x^\top \lambda)$ over $x \in \Delta^{n-1}$ forms a compact convex
polytope.

\begin{proof}
\ For each fixed $x$, $F_\mu(x)$ is affine in $\mu$, hence $g(\mu)$ is a pointwise supremum of affine
functions and therefore convex and continuous on $[0,\infty)$. The attainable set
$S:=\{(x^\top r,x^\top\lambda):x\in\Delta^{n-1}\}$ is the image of a simplex under an affine map, so it is
compact and convex.

Each $\mu\ge0$ defines a supporting hyperplane with normal vector $(1,\mu)$, and $g(\mu)$ is the associated
support function evaluated at that normal direction. Therefore any optimizer $x^\star(\mu)$ attains a
supported efficient point on the boundary of $S$. Since $S$ is a polytope in finite dimension, every
efficient extreme point is supported, hence varying $\mu$ recovers the full efficient frontier up to
degeneracies that correspond to boundary segments.

When the maximizer $x^\star(\mu)$ is unique for a given $\mu$, Danskin's theorem yields
\[
\begin{aligned}
\frac{d}{d\mu} g(\mu)
&=
\frac{\partial}{\partial \mu}
\Bigl(
    x^{\star}(\mu)^{\!\top} r
    - \mu\, x^{\star}(\mu)^{\!\top} \lambda
\Bigr)
\\[4pt]
&=
-\, x^{\star}(\mu)^{\!\top} \lambda.
\end{aligned}
\]
Thus the marginal slope of the efficient frontier with respect to the scalarization weight $\mu$ is
\[
  \frac{d}{d\mu}\, \mathbb{E}[x^{\star}(\mu)^\top r]
  =
  -\, \mathbb{E}[x^{\star}(\mu)^\top \lambda],
\]
The rate at which expected return declines when the investor
tightens the penalty on emissions equals minus the current portfolio intensity. If the maximizer is not
unique, any $x^\star(\mu)$ in the argmax set gives a valid subgradient, and the same interpretation applies
at the level of subdifferentials.
\end{proof}

\subsection{Lipschitz sensitivity in the robustness budget (Theorem~\texorpdfstring{2.3}{2.3})}
\label{sec:appendix-lipschitz}

Recall the robust value function
\[
  R^\star(\Gamma)
  :=
  \max_{x \in \Delta^{n-1}} \,
  \min_{\lVert z \rVert \le \Gamma}
  x^\top \big(\gamma + \Delta z\big),
\]
where $\lVert\cdot\rVert$ is a norm on the disturbance space and $\Gamma > 0$ is the robustness radius. Let
$\lVert\cdot\rVert_\ast$ denote the associated dual norm. Assume that each component return
$r_i(z) := \gamma_i + \delta_i z$ is $L$--Lipschitz in $z$ with respect to $\lVert\cdot\rVert$, so that
\[
  \big|r_i(z_1) - r_i(z_2)\big|
  \le L\, \lVert z_1 - z_2 \rVert,
  \qquad \forall\, z_1,z_2.
\]

\begin{proof} \label{app:lipschitz}
\ Fix $0 < \Gamma_1 < \Gamma_2$ and let $x_2^\star$ and $z_2^\star$ be primal and worst-case disturbance
solutions at robustness $\Gamma_2$:
\[
  R^\star(\Gamma_2)
  =
  x_2^{\star\top}\big(\gamma + \Delta z_2^\star\big),
  \qquad
  \lVert z_2^\star \rVert \le \Gamma_2.
\]
Define the radial projection of $z_2^\star$ onto the smaller ball of radius $\Gamma_1$:
\[
  z_\pi
  :=
  \begin{cases}
    \dfrac{\Gamma_1}{\Gamma_2} z_2^\star, & z_2^\star \neq 0, \\[0.75ex]
    0, & z_2^\star = 0.
  \end{cases}
\]
Then $\lVert z_\pi \rVert \le \Gamma_1$ and
\[
  \lVert z_2^\star - z_\pi \rVert
  \le \Gamma_2 - \Gamma_1.
\]

\paragraph{Justification of the radius bound.}
The inequality follows from positive homogeneity of norms. When $z_2^\star\neq0$,
\[
  \lVert z_2^\star - z_\pi\rVert
  = \Bigl\lVert \Bigl(1-\frac{\Gamma_1}{\Gamma_2}\Bigr) z_2^\star\Bigr\rVert
  = \Bigl(1-\frac{\Gamma_1}{\Gamma_2}\Bigr)\lVert z_2^\star\rVert
  \le \Bigl(1-\frac{\Gamma_1}{\Gamma_2}\Bigr)\Gamma_2
  = \Gamma_2-\Gamma_1.
\]
The case $z_2^\star=0$ is immediate.

By feasibility,
\[
  R^\star(\Gamma_1)
  \,\ge\,
  x_2^{\star\top}\big(\gamma + \Delta z_\pi\big).
\]
Using the Lipschitz property of each $r_i$ and the fact that $x_2^\star \in \Delta^{n-1}$,
\begin{align*}
  x_2^{\star\top}\big(\gamma + \Delta z_\pi\big)
  &\ge
  x_2^{\star\top}\big(\gamma + \Delta z_2^\star\big)
  - L\, \lVert z_2^\star - z_\pi \rVert \\
  &\ge
  R^\star(\Gamma_2) - L (\Gamma_2 - \Gamma_1).
\end{align*}
Hence
\[
  R^\star(\Gamma_1) - R^\star(\Gamma_2)
  \ge -L(\Gamma_2 - \Gamma_1).
\]

Interchanging the roles of $\Gamma_1$ and $\Gamma_2$ yields the reverse inequality, so we obtain the
two-sided bound
\[
  \big|R^\star(\Gamma_1) - R^\star(\Gamma_2)\big|
  \le L\, |\Gamma_2 - \Gamma_1|.
\]
This is the Lipschitz continuity claim in the robustness budget.
\end{proof}

\subsection{Supplementary tables and figures}
\label{sec:appendix-tables}

This subsection collects numerical summaries and graphics referenced in the empirical section.

\input{tables/bootstrap_sharpe_diff}

\input{tables/style_characteristics}

\begin{table}[h!]
  \centering
  \caption{Average 95\% ellipsoidal uncertainty set statistics by calendar year. ``Axis~1'' and
  ``Axis~2'' report the semi-axis lengths of the average covariance ellipse for firm-level emissions
  intensities, ``Area'' is the corresponding ellipse area. Larger areas indicate more dispersion in
  estimated intensities, and hence greater disclosure uncertainty in that year.}
  \label{tab:ellipsoid-stats}
  \begin{tabular}{lccc}
    \toprule
    Year & Axis 1 & Axis 2 & Area \\
    \midrule
    2010 & 0.268 & 0.072 & 0.061 \\
    2011 & 0.281 & 0.082 & 0.074 \\
    2012 & 0.301 & 0.085 & 0.082 \\
    2013 & 0.184 & 0.070 & 0.041 \\
    2014 & 0.167 & 0.058 & 0.031 \\
    2015 & 0.207 & 0.063 & 0.042 \\
    2016 & 0.219 & 0.081 & 0.055 \\
    2017 & 0.136 & 0.083 & 0.035 \\
    2018 & 0.212 & 0.084 & 0.056 \\
    2019 & 0.275 & 0.085 & 0.077 \\
    2020 & 0.475 & 0.158 & 0.236 \\
    2021 & 0.214 & 0.135 & 0.091 \\
    2022 & 0.310 & 0.116 & 0.113 \\
    2023 & 0.197 & 0.110 & 0.068 \\
    2024 & 0.164 & 0.107 & 0.055 \\
    \bottomrule
  \end{tabular}
\end{table}

\begin{table}[h!]
  \centering
  \caption{Top holdings in the baseline EAPO portfolio. ``Avg weight'' reports the mean portfolio
  weight at monthly rebalancing dates over the sample. ``Avg Scope-1 intensity'' reports the
  corresponding time-average of firm-level Scope-1 emissions intensity measured in tCO$_2$e per
  \$mm of revenue. The full monthly weight panels for all strategies are provided as CSV exports in
  the supplementary replication files.}
  \label{tab:eapo-holdings}
  \begin{tabular}{lcc}
    \toprule
    Ticker & Avg weight (\%) & Avg Scope-1 intensity \\
    \midrule
    FLS  & 1.73 & 5.22 \\
    PCAR & 1.73 & 5.52 \\
    BIIB & 1.73 & 5.55 \\
    MCD  & 1.72 & 6.40 \\
    HBAN & 1.71 & 5.02 \\
    HAS  & 1.71 & 6.25 \\
    ED   & 1.70 & 11.82 \\
    EXPD & 1.70 & 5.16 \\
    CAH  & 1.68 & 5.17 \\
    TROW & 1.67 & 4.94 \\
    BMY  & 1.67 & 6.17 \\
    CINF & 1.67 & 6.10 \\
    UPS  & 1.66 & 9.45 \\
    DOV  & 1.65 & 10.83 \\
    JNJ  & 1.65 & 5.50 \\
    \bottomrule
  \end{tabular}
\end{table}

\begin{figure}[h!]
  \centering
  \includegraphics[width=0.95\linewidth]{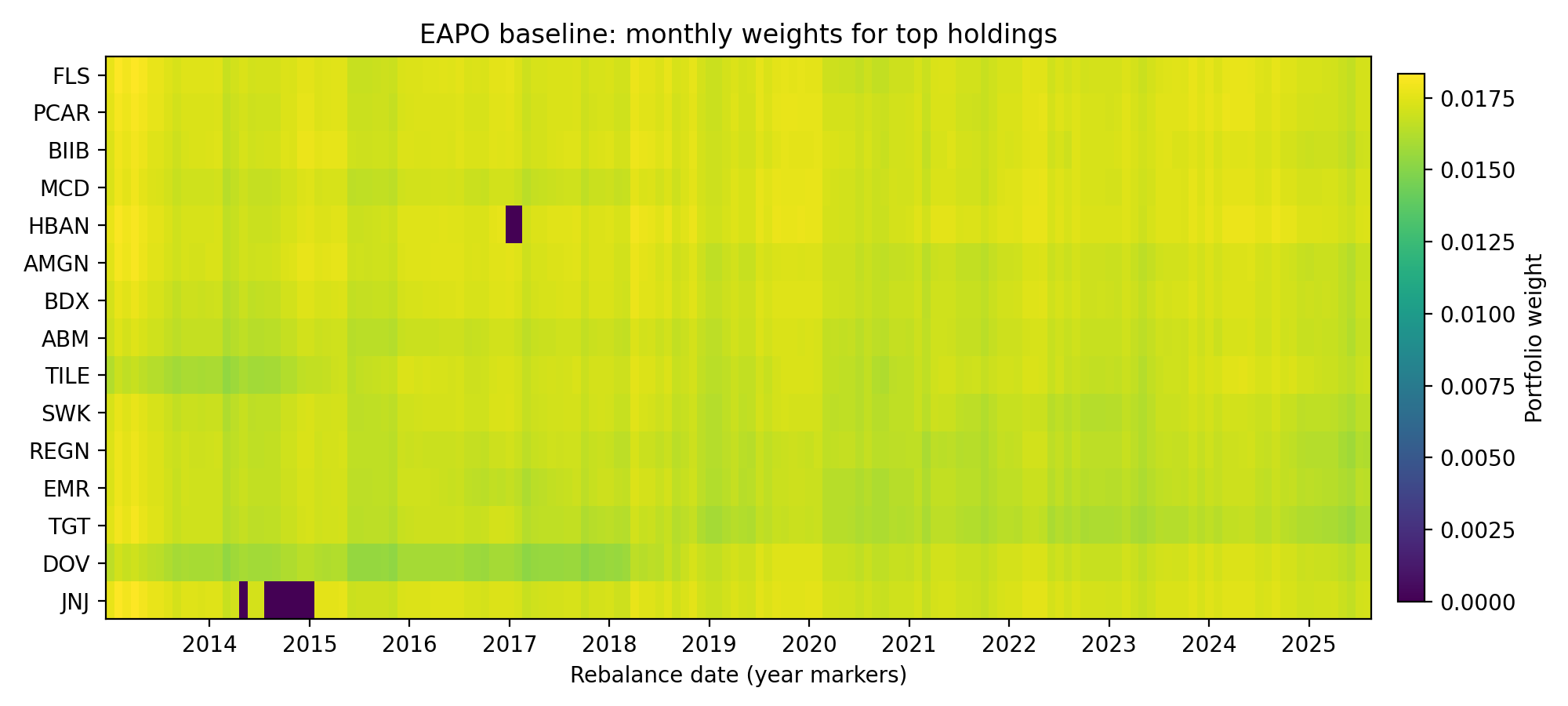}
  \caption{Heatmap of monthly portfolio weights for the top holdings in the baseline EAPO portfolio.
  Each row corresponds to a ticker and each column to a monthly rebalance date. The plot highlights
  that the emissions reductions are achieved through a diversified set of persistent positions rather
  than through episodic concentration in a small set of names.}
  \label{fig:eapo-holdings-heatmap}
\end{figure}



\end{document}

%% file: tables/emissions_attribution.tex
\begin{table}[t]
\centering
\caption{Emissions attribution of EAPO versus the equal-weight benchmark. The decomposition writes the difference in average portfolio intensity as the sum of a sector allocation term and a within-sector selection term.}
\label{tab:emissions_attribution}
\begin{tabular}{l r}
\toprule
Component & Value \\
\midrule
Average Scope-1 intensity, Equal Weight & 246.066 \\
Average Scope-1 intensity, EAPO & 18.297 \\
Total reduction (EW minus EAPO) & 227.769 \\
\midrule
Sector allocation component & 174.112 \ (\ 76.4\%\ ) \\
Within-sector selection component & 53.860 \ (\ 23.6\%\ ) \\
\bottomrule
\end{tabular}
\end{table}

%% file: tables/tracking_diagnostics.tex
\begin{table}[t]
\centering
\caption{Benchmark-tracking diagnostics relative to the equal-weight portfolio. Tracking error is the annualized standard deviation of daily active returns. The information ratio is the annualized mean active return divided by tracking error.}
\label{tab:tracking}
\begin{tabular}{l r r r r}
\toprule
Strategy & Beta & Correlation & Tracking error (\%) & Information ratio \\
\midrule
GMV (inv-var) & 0.864 & 0.975 & 4.464 & -0.874 \\
EMW (1/Scope-1) & 1.008 & 0.883 & 10.047 & 0.150 \\
EAPO & 0.925 & 0.975 & 4.202 & -0.189 \\
\bottomrule
\end{tabular}
\end{table}

%% file: tables/bootstrap_sharpe_diff.tex
\begin{table}[t]
\centering
\caption{Block bootstrap inference for the Sharpe ratio difference between EAPO and the equal-weight benchmark. The bootstrap resamples paired daily returns in blocks of length 20 trading days to preserve short-range dependence.}
\label{tab:bootstrap_sharpe}
\begin{tabular}{l r}
\toprule
Statistic & Value \\
\midrule
Sharpe(EAPO) & 0.762 \\
Sharpe(EW) & 0.766 \\
Difference (EAPO minus EW) & -0.003 \\
\midrule
95\% CI (2.5th, 97.5th) & [ -0.163,\ 0.148 ] \\
Bootstrap replications & 2000 \\
Block length (trading days) & 20 \\
\bottomrule
\end{tabular}
\end{table}

%% file: tables/style_characteristics.tex
\begin{table}[t]
\centering
\caption{Average portfolio exposures to simple price-based characteristics, computed at monthly rebalance dates. Volatility is the weighted-average stock-level daily volatility over the trailing 252 trading days. Momentum is the weighted-average 12--1 momentum, measured as the cumulative return from $t-252$ to $t-21$ trading days.}
\label{tab:style_exposure}
\begin{tabular}{l r r}
\toprule
Strategy & Avg.\ stock vol.\ (\% per day) & Avg.\ 12--1 momentum (\%) \\
\midrule
EW & 1.996 & 15.407 \\
GMV & 1.600 & 13.946 \\
EMW & 2.059 & 17.193 \\
EAPO & 1.844 & 10.730 \\
\bottomrule
\end{tabular}
\end{table}